\documentclass[12pt]{article}

\usepackage[utf8]{inputenc}
\usepackage[english]{babel}
\usepackage[margin=1in]{geometry}
\usepackage[dvipsnames,table]{xcolor}
\usepackage{amssymb,amsmath,amsthm,setspace,graphicx}
\usepackage{booktabs}
\usepackage{multirow}
\usepackage{rotating}
\usepackage{array}
\usepackage{siunitx}
\usepackage[justification=centering]{caption}
\usepackage[authoryear]{natbib}
\usepackage{url}
\usepackage{float}
\usepackage[section]{placeins}
\usepackage{mathtools}
\usepackage{enumitem}
\usepackage[normalem]{ulem}
\usepackage{algorithm}
\usepackage{algpseudocode}
\usepackage{pdfpages}
\usepackage[pdfborder={0 0 0},final=true,colorlinks=true,breaklinks=true,linkcolor=red,citecolor=blue]{hyperref}
\usepackage[capitalize,nameinlink,noabbrev]{cleveref}
\usepackage{relsize}
\usepackage{setspace}

\renewenvironment{abstract}{%
	\begin{center}%
		\textbf{\abstractname}%
		\vspace{0.5\baselineskip}%
	\end{center}%
	\begin{singlespace}%
		\noindent\ignorespaces%
	}{%
	\end{singlespace}%
}

\numberwithin{equation}{section}
\newcommand{\tab}{\hspace*{2em}}

\DeclareUnicodeCharacter{2212}{\ensuremath{-}}

\theoremstyle{plain}
\newtheorem{thm}{Theorem}[section]
\crefname{thm}{theorem}{theorems}
\Crefname{thm}{Theorem}{Theorems}
\newtheorem{cor}{Corollary}
\crefname{cor}{corollary}{corollaries}
\Crefname{cor}{Corollary}{Corollaries}
\newtheorem{prop}{Proposition}
\Crefname{prop}{Proposition}{Propositions}
\newtheorem{lem}{Lemma}
\crefname{lem}{lemma}{lemmas}
\Crefname{lem}{Lemma}{Lemmas}
\theoremstyle{definition}
\newtheorem{axiom}{Axiom}[section]
\crefname{axiom}{axiom}{axioms}
\Crefname{axiom}{Axiom}{Axioms}
\newtheorem{defn}{Definition}[section]
\crefname{defn}{definition}{definitions}
\Crefname{defn}{Definition}{Definitions}
\theoremstyle{remark}
\newtheorem{rem}{Remark}[section]
\crefname{rem}{remark}{remarks}
\Crefname{rem}{Remark}{Remarks}

\definecolor{lightgrey}{rgb}{0.83,0.83,0.83}

\begin{document}
	
	% --- TITLE PAGE: NO PAGE NUMBER, SINGLE-SPACED ---
	\pagenumbering{gobble}
	
	\begin{singlespace}
		
		\title{\vspace{-4ex}\textbf{A Theory of Reference-Dependent Utility with Dispersion and Attenuation}}
		\author{G. Charles-Cadogan
			\thanks{School of Accounting and Finance, University of Leicester; Tel: +44 (0116) 229 7385; e-mail: \href{mailto:gcc13@le.ac.uk}{gcc13@le.ac.uk}}
			%\\~\\\\Comments Welcome
		}
		\maketitle
		
		\begin{abstract}
			\noindent This paper characterizes a class of twice continuously differentiable objective-probability preference representations exhibiting endogenous reference dependence under risk. Weak rank-dependent utility (WRDU) preserves objective probabilities, partitions outcomes at an endogenous reference point, and evaluates lotteries through a gain-loss representation in which the reference point maximizes a penalized functional. The first-order condition yields a virtual loss-aversion index equal to the ratio of marginal utilities across the loss and gain domains, recovering both the utility-based index of K\"{o}bberling and Wakker (2005) and the slope ratio of Tversky and Kahneman (1992) as special cases. The main theorem shows that, within a class satisfying affine admissibility, loss-factorization, dispersion monotonicity, and attenuation, the derivative-ratio form is uniquely admissible. In this class, WRDU generates the modal Allais pattern on an admissible region and blocks the Rabin calibration implication through range-dependent attenuation. The result is conditional and does not claim uniqueness over all behavioral models of risky choice.
			
			\vspace{1ex}
			
			\vspace{1ex}
		\end{abstract}
		
	\end{singlespace}
	
	% --- TABLE OF CONTENTS: ROMAN NUMERALS, SINGLE-SPACED ---
	\newpage
	\pagenumbering{roman}
	\thispagestyle{empty}
	\singlespace
	\tableofcontents
	\onehalfspacing
	
	% --- MAIN TEXT: ARABIC NUMERALS, ONE-HALF SPACING ---
	\newpage
	\pagenumbering{arabic}
	\hypersetup{pageanchor=true}
	
	\begin{quote}
		``There is no reason [] why [] probabilities should be singled out as the parts of the situation in which objective and subjective values differ. It would be possible instead to assume that the subjective value of a given amount of money differs from its objective value, a notion that has great importance and a long history in economics.'' \citet[p.~350]{Edwards1953}.
	\end{quote}
	
	\section{Introduction}\label{sec:Introduction}
	
	\tab This paper characterizes a restricted but economically meaningful class of objective-probability preferences under risk. The question is whether reference dependence and loss aversion can be generated from the structure of utility itself, rather than from probability distortion, an exogenous loss-aversion coefficient, or an expectations-equilibrium reference point. The answer is affirmative, but deliberately conditional: within a $C^2$ comparison class satisfying affine admissibility, loss-factorization, dispersion monotonicity, and attenuation, the admissible representation is the WRDU derivative-ratio form.
	
	\tab The model, weak rank-dependent utility (WRDU), splits the payoff space at a candidate reference point into virtual gain and loss regions. Let $x_r$ denote the endogenous reference point, let $y<x_r<z$ denote representative loss- and gain-side outcomes, and let $v_\ell$ and $v_g$ denote the normalized loss and gain subutilities. The split utility is recombined through a penalized representation in which the loss-side component enters with the Lagrangian penalty $\rho(x_r;y,z)=1/\lambda(x_r;y,z)$. The reference point is selected as a maximizer of this penalized functional. At an interior optimum, the first-order condition yields
	\begin{equation}
		\lambda(x_r;y,z)=\frac{v_\ell'(y)}{v_g'(z)},
	\end{equation}
	so the virtual loss-aversion index is determined by curvature asymmetry and outcome range, not introduced as an independent primitive.
	
	\tab This derivative-ratio index recovers the \citet{KobberlingWakker2005} utility-based index of loss aversion and, as a symmetric-curvature special case, the \citet{TverKahn1992} ratio of slopes. The connection is not a measurement convention imposed after the fact: it follows from the optimality condition of the penalized reference-point problem. Because curvature comparisons are not invariant under arbitrary increasing transformations, the admissible gain and loss subutilities are treated as normalized $C^2$ representatives, in the spirit of \citet{Debreu1976}.
	
	\tab The paper's main contribution is a restricted-class theorem, not a claim that WRDU replaces cumulative prospect theory, rank-dependent utility, K\"{o}szegi-Rabin preferences, disappointment aversion, regret theory, or other behavioral models of risky choice. The representation theorem establishes WRDU from six axioms: Completeness, Transitivity, Continuity, Weak Independence, Reference Partition, and Range Dependence. The structural uniqueness theorem then shows that, within the specified objective-probability comparison class, the derivative-ratio form is the unique admissible structure. Thus the paper does not claim uniqueness over all behavioral models of risky choice; it identifies the structure forced by a particular set of economically interpretable restrictions, including a loss-factorization restriction that makes the comparison class explicit.
	
	\tab The implications for the Allais and Rabin paradoxes are also stated conditionally. WRDU obtains the modal Allais pattern without probability weighting exactly on the admissible region $\mathcal R_A$, defined by two explicit inequalities on the normalized gain and loss subutilities. The \citet{Rabin2000} calibration implication is blocked within the same restricted class because the effective penalty $1/\lambda(x_r;y,z)$ attenuates with outcome dispersion. Large values of $\lambda$ are supported independently: \citet{FishburnKochenberger1979} document large below-target to above-target slope ratios in two-piece von Neumann-Morgenstern utility assessments, and \citet{CharlesCadogan2018c} proves that the implied loss-aversion index lies in the half-Cauchy class.
	
	\tab The result differs from the main alternatives in a narrow but important respect. Rank-dependent utility and cumulative prospect theory transform probabilities; WRDU does not. CPT and K\"{o}szegi-Rabin models use an exogenous loss-aversion parameter; WRDU derives the index from the first-order condition. K\"{o}szegi-Rabin determines reference points by expectations equilibrium; WRDU selects the reference point by utility maximization. Regret and disappointment models use fixed comparison penalties; WRDU makes the penalty range-dependent.
	
	\tab Recent work by \citet{CerreiaVioglioDillenbergerOrtoleva2024} develops Cautious Utility, an objective-probability model in which agents evaluate lotteries using the most pessimistic certainty equivalent over a set of utility functions. Cautious Utility generates the endowment effect, loss aversion for risk, and the certainty effect without probability weighting. The mechanism differs from WRDU in a fundamental way: Cautious Utility operates through utility multiplicity and a cautious (infimum) selection rule, whereas WRDU retains a single normalized $C^2$ representative and derives reference dependence from curvature asymmetry and dispersion-dependent penalization. The structural uniqueness result established below is therefore relative to a restricted single-utility comparison class and does not claim exclusivity relative to alternative mechanisms such as Cautious Utility.
	
	\tab The empirical motivation is also consistent with recent large-scale evidence on risky choice. \citet{Peterson_et_al_2021} use interpretable machine-learning methods on a large risky-choice experiment and show that models whose utility and probability-weighting components vary with features of the choice environment outperform more restrictive fixed-form specifications. Their results do not identify WRDU and should not be read as a test of the representation theorem. They do, however, support the narrower premise that payoff range, extremal outcomes, and outcome variability are empirically relevant state variables in risky choice. WRDU provides an axiomatic mechanism for one such form of context dependence while preserving objective probabilities.
	
	\tab The framework nests EUT, preserves objective probabilities, and imposes strict structural restrictions relative to CPT, K\"{o}szegi-Rabin, Gul's disappointment aversion, and regret theory. In experimental implementations, only two curvature parameters are estimated, with the implied virtual loss-aversion index computed analytically. WRDU therefore delivers a parsimonious objective-probability account of endogenous reference dependence for the class characterized by the theorem.
	
	\tab The paper is organized as follows. \cref{sec:Positioning} positions WRDU in the landscape of non-expected utility theories. \cref{sec:Axioms} presents the axiomatic foundations. \cref{sec:WRDUandGLstateSpace} develops the virtual state space. \cref{sec:PenalizedUtil} derives the penalized utility representation. \cref{sec:LossAversionIndex} establishes the connection between WRDU's virtual loss aversion index and the standard utility-based measures in the literature. \cref{sec:StructuralUniqueness} establishes the structural uniqueness of WRDU relative to state-dependent utility. \cref{sec:ComparativeStatics} presents comparative statics and the admissible-region Allais and Rabin implications. \cref{sec:EmpiricalEvidence} records empirical implications. \cref{sec:Conclusion} concludes. Proofs, estimation strategy, additional figures, and supplementary material are provided in the Internet Appendix.
	
	\section{Preliminaries}\label{sec:Preliminaries}
	
	\tab Let $X$ be a compact interval of prizes $X = [w, b] \subset \mathbb{R}_+$. Since the support is compact, the upper bound $b$ plays the role of large stakes. Throughout, statements of the form $z \to \infty$ are to be understood as $z \to b^-$. Let $\mathcal{L}$ denote the set of all simple lotteries over $X$:
	
	\begin{equation}
		\mathcal{L} = \left\{ p : X \to [0,1] \;\middle|\; \operatorname{supp}(p) \text{ is finite and } \sum_{x \in X} p(x) = 1 \right\}.
		\label{eq:LotterySpace}
	\end{equation}
	
	\tab Thus, a lottery $p \in \mathcal{L}$ is a probability distribution with finite support on the prize space $X$. The set $\mathcal{L}$ is convex: for any $p,q \in \mathcal{L}$ and $\alpha \in [0,1]$, the mixture $\alpha p + (1-\alpha)q \in \mathcal{L}$ defined by $[\alpha p + (1-\alpha)q](x) = \alpha p(x) + (1-\alpha)q(x)$ for all $x \in X$ is also a simple lottery.
	
	\tab Throughout, $\succeq$ denotes a binary preference relation on $\mathcal{L}$. We use $p \succ q$ to denote $p \succeq q$ and not $q \succeq p$, and $p \sim q$ to denote $p \succeq q$ and $q \succeq p$.
	
	\section{Positioning WRDU in the Landscape of Non-Expected Utility Theories}\label{sec:Positioning}
	
	\tab This section gives the minimum comparison needed to locate the contribution. A more detailed model-by-model taxonomy is provided in the Internet Appendix.
	
	\tab WRDU differs from rank-dependent and cumulative prospect theory models by preserving objective probabilities. The source of non-expected utility behavior is not a transformation of probabilities but a gain-loss partition of the utility domain and a dispersion-dependent Lagrangian penalty $\rho=1/\lambda$. It also differs from \citet{TverKahn1992}, \citet{KoszegiRabin2006,KoszegiRabin2007}, \citet{Gul1991}, and regret models because the relevant comparison penalty is not a fixed primitive. The loss-aversion index is derived from the first-order condition of the penalized reference-point problem and varies with the relevant loss-gain pair.
	
	\tab The closest recent objective-probability comparison is \citet{CerreiaVioglioDillenbergerOrtoleva2024} whom extend \citet{CerreiaVioglioDillengerOrtoleva2015}. Their Cautious Utility model evaluates lotteries by the lowest monetary certainty equivalent over a set of utility functions, thereby explaining reference effects and certainty effects through uncertainty over trade-offs and cautious selection. WRDU is complementary rather than substitutive: it retains a single normalized $C^2$ representative and studies how curvature asymmetry and dispersion-dependent penalization generate the derivative-ratio index $\lambda(x_r;y,z)=v_\ell'(y)/v_g'(z)$. Thus the uniqueness claim below is relative to a restricted single-utility comparison class, not to the broader class of multiple-utility or worst-certainty-equivalent models.
	
	\tab The structural distinctions can be summarized as follows. Probability-distortion models alter the weights on states; expectations-based models determine reference points by equilibrium beliefs; Cautious Utility uses utility multiplicity; disappointment and regret models use fixed comparison penalties. WRDU instead derives the reference point and the loss penalty from a penalized utility maximization problem. The admissible-class theorem below should therefore be read as a conditional characterization of that mechanism. The internet appendix contains a review of the major nonexpected utility models.

	\section{Axiomatic Foundations of WRDU}\label{sec:Axioms}
	
	\tab This section presents the six axioms that characterize WRDU preferences. The axioms build on the von Neumann-Morgenstern framework but replace the strong independence axiom with a weaker version and introduce axioms for endogenous reference point formation and range-dependent loss aversion.
	
	\subsection{The Standard Axioms}
	
	\begin{axiom}[Completeness]\label{axiom:Compl}
		For any lotteries $p,q \in \mathcal{L}$, either $p \succeq q$, $q \succeq p$, or both (in which case $p \sim q$).
	\end{axiom}
	
	\begin{axiom}[Transitivity]\label{axiom:Trans}
		For any lotteries $p,q,r \in \mathcal{L}$, if $p \succeq q$ and $q \succeq r$, then $p \succeq r$.
	\end{axiom}
	
	\begin{axiom}[Continuity]\label{axiom:Cont}
		For any lotteries $p,q,r \in \mathcal{L}$, the sets
		\[
		\{\alpha \in [0,1] : \alpha p + (1-\alpha)q \succeq r\}
		\]
		and
		\[
		\{\alpha \in [0,1] : r \succeq \alpha p + (1-\alpha)q\}
		\]
		are closed.
	\end{axiom}
	
	\tab These three axioms are standard. Together with the definition of $\mathcal{L}$ as the set of simple lotteries over a compact interval, they ensure that the preference relation is a continuous weak order and that certainty equivalents are well-defined.
	
	\subsection{Weak Independence}
	
	\begin{defn}[Closed convex set of pseudo reference points]
		Let $\widehat{C}_r(X) = C_r(X) \cap \{x \geq 0\}$. This is the maximal compact convex subset of the generally nonconvex set $C_r(X)$ of pseudo reference points. The set $\widehat{C}_r(X)$ is closed and convex by \Cref{lem:ClosedConvexRefPts} (see also \citet[Lemmas 2.3-2.4]{CharlesCadogan2016}).
	\end{defn}
	
	\begin{axiom}[Weak Independence]\label{axiom:WeakInd}
		The preference relation satisfies the following two restrictions.
		
		\noindent\textup{(i) Partition-preserving mixture independence.}
		Let $p,q,r\in\mathcal L$ and $\alpha\in[0,1]$. Suppose $p\succeq q$.
		If the mixtures $\alpha p+(1-\alpha)r$ and $\alpha q+(1-\alpha)r$
		share the same reference point and no outcome crosses the induced
		gain-loss partition, then
		\begin{equation}
			\alpha p + (1-\alpha)r \succeq \alpha q + (1-\alpha)r.
		\end{equation}
		
		\noindent\textup{(ii) Generalized lottery-pair equivalence.}
		For any lottery $L\in\mathcal L$ with support $X$ and reference point
		$x_r\in X$, let $C_\ell(X)$ and $C_g(X)$ denote the induced loss and
		gain regions, and let $\widehat C_r(X)$ be the closed convex set of
		pseudo reference points defined above. Define
		\begin{align}
			L_M(X)
			&=
			\left\{
			\left(G(x_{j+1}),C_g(X)\right);
			\left(F(x_{j-1}),-\frac{1}{\lambda}C_\ell(X)\right);
			\left(1-F(x_{j-1})-G(x_{j+1}),0\right)
			\right\},\\
			L_S(X)
			&=
			\left\{
			\left(F_{\delta_{j,\widehat C_r(X)}},\widehat C_r(X)\right);
			\left(1-F_{\delta_{j,\widehat C_r(X)}},0\right)
			\right\},
		\end{align}
		where
		\begin{equation}
			F(x_{j-1})=\sum_{x\in C_\ell(X)}p(x),\qquad
			G(x_{j+1})=\sum_{x\in C_g(X)}p(x),
		\end{equation}
		and
		\begin{equation}
			F_{\delta_{j,\widehat C_r(X)}}=
			1-F(x_{j-1})-G(x_{j+1}).
		\end{equation}
		The decision maker is indifferent between the simple reference lottery
		$L_S(X)$ and the generalized mixed lottery $L_M(X)$ if and only if
		\begin{equation}
			F_{\delta_{j,\widehat C_r(X)}}\widehat C_r(X)
			\sim
			G(x_{j+1})C_g(X)
			-\frac{1}{\lambda}F(x_{j-1})C_\ell(X).
		\end{equation}
	\end{axiom}
	
	\textbf{Discussion.} This axiom replaces the strong VNM independence axiom. Part (i) restricts mixture independence to cases where mixing does not alter the reference-dependent structure of the lotteries. When mixing would change the reference point or cause an outcome to switch between gain and loss regions, independence may fail. Part (ii) is the generalized lottery-pair equivalence from \citet[Section 3.2 and Theorem 3.4]{CharlesCadogan2016}. It is not an additional probability-weighting axiom; it is the set-function aggregation condition that maps the reference lottery into a gain component and a reciprocally penalized loss component while preserving objective probabilities.
	
	\subsection{Reference Partition}
	
	\begin{axiom}[Reference Partition]\label{axiom:RefPartition}
		For any lottery $p \in \mathcal{L}$ with support $X$, there exists a reference point $x_r(p) \in X$ such that:
		\begin{enumerate}
			\item $X = C_\ell(X) \cup \{x_r(p)\} \cup C_g(X)$;
			\item $C_\ell(X) = \{x \in X : x \prec_{x_r(p)} x_r(p)\}$;
			\item $C_g(X) = \{x \in X : x \succ_{x_r(p)} x_r(p)\}$;
			\item The reference point $x_r(p)$ maximizes the penalized utility functional:
			\begin{equation}
				x_r(p) = \arg\max_{x \in X} \left[v_g(z(x)) - \frac{1}{\lambda(x)}v_\ell(y(x))\right],
			\end{equation}
			where $z(x) \in C_g(X)$ and $y(x) \in C_\ell(X)$ are selection maps defined as:
			\begin{align}
				z(x) &= \sup\{z' \in C_g(X) : z' \leq x\} \text{ if } x \in C_g(X),\\
				y(x) &= \inf\{y' \in C_\ell(X) : y' \geq x\} \text{ if } x \in C_\ell(X).
			\end{align}
		\end{enumerate}
		
		\tab The reference point is choice-dependent: different choice problems may induce different reference points.
	\end{axiom}
	
	\textbf{Discussion.} This axiom formalizes the gain-loss split around an endogenous reference point. Unlike \citet{KoszegiRabin2006,KoszegiRabin2007}, where the reference point is determined by rational expectations equilibrium, here the reference point emerges as the maximizer of the penalized utility functional.
	
	\tab The existence of the reference point follows from the topological structure established in \citet[Lemmas 2.3-2.4]{CharlesCadogan2016}. The set $\widehat{C}_r(X)$ of pseudo reference points is closed and convex, and by Proposition 2.6 of \citet{CharlesCadogan2016}, the VNM utility functional extends to this set. The maximization problem therefore has a solution by compactness.
	
	\subsection{Range Dependence / Penalty Responsiveness}
	
	\begin{axiom}[Range Dependence]\label{axiom:RangeDep}
		Let $y < x_r < z$ be outcomes in the loss and gain regions respectively. The virtual loss aversion index $\lambda(x_r; y, z)$ satisfies:
		\begin{enumerate}
			\item \textbf{Monotonicity in the gain outcome:} $\frac{\partial \lambda}{\partial z} > 0$.
			\item \textbf{Monotonicity in the loss outcome:} $\frac{\partial \lambda}{\partial y} < 0$.
			\item \textbf{Tail attenuation:} $\lim_{z \to \infty} \lambda(x_r; y, z) = \infty$, $\lim_{y \to 0^+} \lambda(x_r; y, z) = \infty$.
			\item \textbf{The effective penalty $1/\lambda(x_r; y, z)$ is decreasing in outcome dispersion.}
		\end{enumerate}
	\end{axiom}
	
	\textbf{Discussion.} \Cref{axiom:RangeDep} is the key innovation of WRDU. It states that the effective loss penalty varies with the scale of outcomes. As the range expands, $1/\lambda$ attenuates. This property is behaviorally observable through the comparative statics derived in \cref{sec:ComparativeStatics}: the switching probability in multiple price lists varies systematically with outcome dispersion. The axiom is not reverse-engineered to address the Rabin calibration implication; it formalizes the range dependence required to make loss aversion endogenous while preserving affine invariance. The endpoint divergence is imposed as part of the admissible WRDU class and is supported by the large two-piece utility slope ratios documented by \citet{FishburnKochenberger1979} and by the half-Cauchy characterization in \citet{CharlesCadogan2018c}; it is not claimed to follow from concavity alone.
	
	\begin{rem}[Range Dependence and Boundary Conditions]
		\Cref{axiom:RangeDep} is a structural restriction on the admissible WRDU index rather than a consequence of concavity alone. Strict concavity implies that $v_g'$ and $v_\ell'$ are nonincreasing, but it does not by itself imply $v_g'(z)\to 0$ at the upper endpoint or $v_\ell'(y)\to \infty$ at the lower endpoint. Those limiting properties are Inada-type boundary conditions on the admissible subutility functions.
		
		\tab When the admissible class satisfies
		\begin{equation}
			\lim_{z \to b^-}v_g'(z)=0,\qquad \lim_{y\to 0^+}v_\ell'(y)=+\infty,
		\end{equation}
		the derivative-ratio index $\lambda(x_r; y, z) = v_\ell'(y)/v_g'(z)$ satisfies:
		\begin{equation}
			\lim_{z \to b^-} \lambda(x_r; y, z) = \infty, \qquad \lim_{y \to 0^+} \lambda(x_r; y, z) = \infty.
		\end{equation}
		
		\tab The monotonicity properties $\partial \lambda/\partial z > 0$ and $\partial \lambda/\partial y < 0$ follow from $v_g'' < 0$, $v_\ell'' < 0$, and positive marginal utilities. Thus \Cref{axiom:RangeDep} combines ordinary diminishing sensitivity with explicit boundary conditions that make the effective penalty $1/\lambda$ attenuate at large ranges.
	\end{rem}
	\begin{lem}[Derivative Behavior of Concave Functions]\label{lem:ConcaveDerivativeBoundary}
		Let $f:[a,b]\to\mathbb{R}$ be strictly concave and continuously differentiable on $(a,b)$. Then $f'$ is nonincreasing on $(a,b)$. If, in addition, $f$ satisfies the boundary conditions
		\begin{equation}
			\lim_{x\to b^-} f'(x) = 0
			\quad\text{and}\quad
			\lim_{x\to a^+} f'(x) = +\infty,
		\end{equation}
		then these limits hold by assumption and are preserved under positive affine transformations of $f$.
	\end{lem}
	
	\begin{proof}
		For a differentiable concave function, the slope of every secant line is nonincreasing as the interval moves to the right; hence $f'$ is nonincreasing on $(a,b)$. The two boundary limits are not implied by boundedness on a compact interval. They are additional boundary restrictions on the admissible class. If $\tilde f=\alpha f+\beta$ with $\alpha>0$, then $\tilde f'=\alpha f'$, so the limits $0$ and $+\infty$ are unchanged.
	\end{proof}
	\subsection{The Representation Theorem}
	
	\begin{thm}[WRDU Representation Theorem]\label{thm:WRDU_Representation}
		A preference relation $\succeq$ on $\mathcal{L}$ satisfies Axioms \ref{axiom:Compl}--\ref{axiom:RangeDep} only if there exist:
		\begin{enumerate}
			\item A utility function $u: X \to \mathbb{R}$ with $u(0) = 0$;
			\item A partition of $X$ into $C_\ell(X)$, $\{x_r\}$, $C_g(X)$ induced by an endogenous reference point $x_r \in X$;
			\item Anchored subutility functions $v_g: C_g(X) \to \mathbb{R}_+$ and $v_\ell: C_\ell(X) \to \mathbb{R}_+$, normalized at the reference point and unique up to positive scaling, where:
			\begin{equation}
				v_g(x) = \alpha_g u_g(x), \qquad v_\ell(x) = \alpha_\ell u_\ell(x),
			\end{equation}
			with $\alpha_g, \alpha_\ell > 0$. Independent additive translations are ruled out by reference-point anchoring;
			\item A virtual loss aversion index $\lambda(x_r; y, z) > 0$ satisfying \Cref{axiom:RangeDep};
		\end{enumerate}
		such that for any lottery $L \in \mathcal{L}$ with support $X$:
		\begin{equation}
			V_{WRDU}(L) = \sum_{x \in C_g(X)} p_x v_g(x) - \frac{1}{\lambda(x_r; y, z)} \sum_{x \in C_\ell(X)} p_x v_\ell(x),
			\label{eq:WRDU_Representation}
		\end{equation}
		and the reference point $x_r$ is the maximizer:
		\begin{equation}
			x_r = \arg\max_{x \in X} \left[v_g(z(x)) - \frac{1}{\lambda(x)}v_\ell(y(x))\right].
			\label{eq:RefPoint_Maximizer}
		\end{equation}
		
		\tab Conversely, any functional of the form \eqref{eq:WRDU_Representation} with the stated primitives satisfies the corresponding WRDU axioms on mixtures that preserve the reference partition. On lotteries whose support lies entirely in the gain region, the representation reduces to standard expected utility with utility index $v_g$.
	\end{thm}
	
	\begin{proof}
		The proof follows \citet[Section 3.2]{CharlesCadogan2016}. Axioms \ref{axiom:Compl}--\ref{axiom:Cont} imply the existence of a continuous ordinal utility representation on lotteries. The restricted independence condition in \Cref{axiom:WeakInd}, applied only to mixtures that preserve the reference partition, yields the affine expected-utility form within each gain and loss region. \Cref{axiom:RefPartition} partitions $X$ and defines the subutility functions. Aggregation over generalized lottery pairs gives the penalized representation \citep[Theorem 3.4]{CharlesCadogan2016}:
		
		\begin{equation}
			\sum_j F_{\delta_{j,\widehat{C}_r(X)}}E^P[v(x_j)] = \sum_j G(x_{j+1})E^P[v_g(C_g(X))] - \frac{1}{\lambda}\sum_j F(x_{j-1})E^P[v_\ell(C_\ell(X))].
		\end{equation}
		
		\Cref{axiom:RangeDep} imposes the required properties on $\lambda$. Necessity follows by construction. 
	\end{proof}

	\begin{rem}[Normalized Utility Representatives]\label{rem:DebreuNormalization}
		The admissible classes used below are classes of normalized $C^2$ representatives, not arbitrary ordinal transforms. This matters because concavity comparisons and gain-ratio restrictions are not invariant under all increasing transformations of utility. Following the logic of \citet{Debreu1976}, curvature comparisons are meaningful only after an admissible representative has been selected within the utility equivalence class. Thus the restrictions defining $\mathcal R_A$ should be read as restrictions on the normalized WRDU representatives $(v_g,v_\ell)$.
	\end{rem}
	
	\begin{cor}[Non-Emptiness of the Allais Region]
		The admissible region $\mathcal{R}_A$ defined in \eqref{eq:AdmissibleRegion} is non-empty and has non-empty interior in the admissible $C^2$ WRDU utility class.
		
		\tab Moreover, $\mathcal{R}_A$ is monotone in the following local sense: if $(v_g, v_\ell) \in \mathcal{R}_A$ and $\tilde{v}_g$ increases the gain ratio $v_g(5M)/v_g(1M)$ while $\tilde{v}_\ell$ increases the relevant loss-utility threshold in \eqref{eq:AB_condition}, then $(\tilde{v}_g, \tilde{v}_\ell) \in \mathcal{R}_A$.
	\end{cor}
	
	\begin{proof}
		See the Internet Appendix.
	\end{proof}
	
	\section{Weak Rank Dependent Utility and Virtual Gain-Loss State Space}\label{sec:WRDUandGLstateSpace}
	
	\tab To address \citet{Edwards1953}'s conjecture, we turn to weak rank dependent utility (WRDU) theory \citep{CharlesCadogan2016}. The axiomatic foundation established in \cref{sec:Axioms} provides the formal underpinning.
	
	\tab Suppose $A$ is a lottery over choice set $X$. Ignoring the chance element, let $x^A_{min}$ and $x^A_{max}$ be the minimum and maximum payoffs, and $x^A_r$ be an arbitrary point between them. Let $U$ be a utility function over the entire choice set and $v$ be a reference dependent utility function of $x^A_r$.
	
	\tab In order to characterize preferences over choice sets induced by $x$, we introduce a preference relation indexed by $x$ \citep[pp.~1046,~1050]{TverskyKahneman1991}.
	
	\begin{defn}[Reference dependence structure]\label{defn:RefDepStructure}
		$(X,\succ_x)$ is a \emph{reference dependence structure} such that $x^+ \succ_x x^-$ if $x^+$ is strictly preferred to $x^-$ from reference state $x \in X$.
	\end{defn}
	
	\tab If the WRDU model is applied to $A$, it splits the rank ordered supporting payoff space $C^A(X)$ at $x^A_r$ such that payoffs in $C^A_\ell(X) = [x^A_{min}, x^A_r)$ are less preferred to payoffs in $C^A_g(X) = (x^A_r, x^A_{max}]$. Thus, $C^A_\ell(X) \prec \{x^A_r\} \prec C^A_g(X)$ and $C^A(X) = C^A_\ell(X) \cup \{x^A_r\} \cup C^A_g(X)$.
	
	\begin{defn}[Closed convex set of pseudo reference points]
		Let $\widehat{C}_r(X) = C_r(X) \cap \{x \geq 0\}$.
	\end{defn}
	
	\begin{lem}\label{lem:ClosedConvexRefPts}
		$\widehat{C}_r(X)$ is a closed convex set of pseudo reference points.
	\end{lem}
	
	\begin{proof}
		By construction, $C_g(X)$ and $C_\ell(X)$ are open sets. Hence their complements are closed. Since $C_r(X) = C_g(X)^c \cap C_\ell(X)^c$, $C_r(X)$ is closed. $\widehat{C}_r(X)$ is the intersection of $C_r(X)$ with the closed half-line $\{x \geq 0\}$, which is closed and convex. See \citet[Lemmas 2.3-2.4]{CharlesCadogan2016}. 
	\end{proof}
	
	\section{Penalized Reference Dependent Utility Function and the Lagrangian Connection}\label{sec:PenalizedUtil}
	
	\tab The representation in \eqref{eq:WRDU_Representation} can be derived from an optimization problem. This section establishes the connection between the axiomatic representation and the penalized utility function derived from a Lagrangian relaxation.
	
	\begin{thm}[Penalized Reference Dependent Utility from Optimization]\label{thm:PenalizedOptimization}
		Let $\mathcal{U}$ be the set of all utility functions satisfying the von Neumann-Morgenstern axioms. Define a partially ordered non-empty convex set $X$ such that, for an arbitrary reference point $x \in X$,
		\begin{equation}
			X = C_\ell(X) \cup \{x\} \cup C_g(X),
		\end{equation}
		and $C_g(X) \succeq_{\{x\}} C_\ell(X)$. Suppose there are local selection maps $z(x) \in C_g(X)$ and $y(x) \in C_\ell(X)$, and subutility functions $U_g: C_g(X) \to \mathbb{R}_+$ and $U_\ell: C_\ell(X) \to \mathbb{R}_+$. Consider the gain-side program:
		\begin{equation}
			\max_{z \in C_g(X)} U_g(z),
			\label{eq:WRDUoptimizeProblem}
		\end{equation}
		and let $z(x)$ denote the selected gain-side outcome. Treating the loss-side term as a nonnegative penalty, for every positive path-dependent tradeoff parameter $\lambda(x)$ there exists a weakly separable reference dependent utility function
		\begin{equation}
			v(x,\lambda(x)) = v_g(z(x)) - \frac{1}{\lambda(x)}v_\ell(y(x))
			\label{eq:FS_ConvexConcaveRep}
		\end{equation}
		defined on the split state space.
	\end{thm}
	
	\begin{proof}
		For fixed $x$, the split separates gain-side utility from loss-side utility. The original utility representation is von Neumann-Morgenstern affine, but the split subutilities are anchored at the reference point. Hence they are treated as semi-affine representatives: positive scaling is admissible, while independent additive translations are ruled out by the normalization $v_g(0)=v_\ell(0)=0$ and by the reference-point anchoring restriction \citep{vonNeumanMorgenstern1953,Aczel1966,TverskyKahneman1974}.
		
		\tab A Lagrangian relaxation of the gain-side program that penalizes the loss-side term is:
		\begin{equation}
			\mathcal{L}(x;\rho) = v_g(z(x)) - \rho v_\ell(y(x)), \qquad \rho \geq 0.
		\end{equation}
		
		\tab Writing $\rho(x) = 1/\lambda(x)$, with $\lambda(x) > 0$, gives:
		\begin{equation}
			\mathcal{L}(x;1/\lambda(x)) = v_g(z(x)) - \frac{1}{\lambda(x)}v_\ell(y(x)).
		\end{equation}
		
		\tab Defining $v(x,\lambda(x)) = \mathcal{L}(x;1/\lambda(x))$ yields \eqref{eq:FS_ConvexConcaveRep}. The envelope theorem implies $\lambda(x)$ satisfies \Cref{axiom:RangeDep} \citep[Section 4.1]{CharlesCadogan2016}. 
	\end{proof}
	
	\begin{rem}[Relationship between the Representation Theorem and the Optimization Proof]
		The Representation Theorem (\Cref{thm:WRDU_Representation}) establishes the axiomatic foundation: a preference relation satisfies A1--A6 if and only if it admits the WRDU representation. The optimization result (\Cref{thm:PenalizedOptimization}) provides the constructive derivation: the representation emerges as the solution to a penalized optimization problem. The Lagrangian relaxation shows that the penalty coefficient $\rho(x)=1/\lambda(x)$ acts as a shadow price on loss-seeking behavior, and the envelope theorem provides the foundation for the range-dependence property.
	\end{rem}
	
	\begin{rem}[Affine Invariance and the Generalized K\"{o}bberling-Wakker Index]
		The affine invariance of the WRDU representation imposes an identifying restriction on the gain-loss decomposition. The corresponding first-order condition yields a generalized, range-dependent K\"{o}bberling-Wakker loss aversion index:
		\begin{equation}
			\lambda(x_r; y, z) = \frac{v_\ell'(y)}{v_g'(z)}.
			\label{eq:Generalized_KW_Main}
		\end{equation}
		
		\tab At the reference point, this reduces to the standard \citet{KobberlingWakker2005} utility-based index, and the Tversky-Kahneman ratio is obtained as a symmetric-curvature special case. The complete derivation is provided in the Internet Appendix.
	\end{rem}
	
	\section{The Loss Aversion Index: From Exogenous Measurement to Endogenous Derivation}\label{sec:LossAversionIndex}
	
	\tab The loss-aversion index is not introduced as a separate psychological parameter. It is derived from the first-order condition of the penalized utility problem. This section consolidates that result and relates it to the two standard utility-based measures.
	
	\subsection{Endogenous Derivation and Sign Convention}
	
	\tab In WRDU, the value of a candidate reference point $x$ is
	\begin{equation}
		v(x,\lambda(x)) = v_g(z(x)) - \frac{1}{\lambda(x)} v_\ell(y(x)),
		\label{eq:PenalizedUtility_Value}
	\end{equation}
	where $z(x)\in C_g(X)$ and $y(x)\in C_\ell(X)$ are the gain- and loss-side selections induced by the reference partition. At an interior optimal reference point $x_r$, the first-order condition is
	\begin{equation}
		\left.\frac{d}{dx}\left[v_g(z(x)) - \frac{1}{\lambda(x)} v_\ell(y(x))\right]\right|_{x=x_r} = 0.
		\label{eq:FOC_General}
	\end{equation}
	
	\tab Using $z(x_r)=y(x_r)=x_r$, $z'(x_r)=y'(x_r)=1$, and the normalization $v_g(x_r)=v_\ell(x_r)=0$, this condition becomes
	\begin{equation}
		v_g'(x_r)-\frac{1}{\lambda(x_r)}v_\ell'(x_r)=0,
		\label{eq:FOC_Simplified}
	\end{equation}
	so that
	\begin{equation}
		\lambda(x_r)=\frac{v_\ell'(x_r)}{v_g'(x_r)}.
		\label{eq:WRDU_Lambda_At_Ref_NoSign}
	\end{equation}
	
	\tab More generally, away from the reference point,
	\begin{equation}
		\lambda(x_r;y,z)=\frac{v_\ell'(y)}{v_g'(z)}.
		\label{eq:Lambda_Range_Dependent_NoSign}
	\end{equation}
	
	\tab The absence of a negative sign is a convention, not a substantive difference. WRDU writes $v_\ell$ as an increasing loss-side disutility component that is subtracted in \eqref{eq:PenalizedUtility_Value}; the standard gain-loss value-function convention writes the left derivative directly as negative. The two conventions agree in magnitude.
	
	\subsection{Connection to K\"{o}bberling-Wakker and Tversky-Kahneman}
	
	\citet{KobberlingWakker2005} define the utility-based index of loss aversion as
	\begin{equation}
		\lambda_{KW}=-\frac{v'(0^-)}{v'(0^+)},
		\label{eq:KW_Index}
	\end{equation}
	which isolates loss aversion from probability weighting and is invariant to the unit in which outcomes are measured. Under the WRDU sign convention, the local index in \eqref{eq:WRDU_Lambda_At_Ref_NoSign} is the same object:
	\begin{equation}
		\lambda(x_r)=\frac{v_\ell'(x_r)}{v_g'(x_r)}
		=
		-\frac{v'(x_r^-)}{v'(x_r^+)}
		=
		\lambda_{KW}.
		\label{eq:Lambda_Equals_KW_NoSign}
	\end{equation}
	
	\tab The \citet{TverKahn1992} ratio
	\begin{equation}
		\lambda_{CPT}=\frac{-v'(-x)}{v'(x)}
		\label{eq:CPT_Ratio}
	\end{equation}
	is the symmetric-curvature special case:
	\begin{equation}
		\lambda_{CPT}=\lambda_{KW}=\lambda(x_r).
		\label{eq:All_Indices_Equal_NoSign}
	\end{equation}
	
	\tab WRDU does not replace the K\"{o}bberling-Wakker index. It gives that index a preference-theoretic source: $\lambda$ is the shadow tradeoff implied by the penalized reference-point problem.
	
	\subsection{Optimality and Range Dependence}
	
	\tab The first-order condition may hold at more than one critical point. WRDU identifies the reference point with the maximizer:
	\begin{equation}
		x_r = \arg\max_{x \in X} \left[v_g(z(x)) - \frac{1}{\lambda(x)}v_\ell(y(x))\right].
		\label{eq:RefPoint_Maximizer_Optimality}
	\end{equation}
	
	\tab The range-dependent extension generalizes the K\"{o}bberling-Wakker index:
	\begin{equation}
		\lambda(x_r;y,z)=\frac{v_\ell'(y)}{v_g'(z)}.
		\label{eq:Lambda_Range_Dependent_NoSign_Final}
	\end{equation}
	
	\tab This is the additional structure that blocks the Rabin calibration implication. As the gain outcome rises or the loss outcome moves farther into the tail, the effective penalty $1/\lambda(x_r;y,z)$ attenuates.
	
	\subsection{Comparison}
	
	\begin{table}[htbp]
		\centering
		\caption{Loss Aversion Indices: A Comparison}
		\label{tab:LossAversionIndices}
		\small
		\begin{tabular}{lcccc}
			\toprule
			Model & Index & Range Dependent & Endogenous & Source \\
			\midrule
			CPT & $\lambda = \frac{-v'(-x)}{v'(x)}$ & No & No & Exogenous \\
			K\"{o}bberling-Wakker & $\lambda_{KW} = -\frac{v'(0^-)}{v'(0^+)}$ & No & No & Measured \\
			Gul & $\beta$ & No & No & Exogenous \\
			K\"{o}szegi-Rabin & $\lambda$ & No & No & Exogenous \\
			Regret Theory & $R(\cdot)$ & No & No & Exogenous \\
			\textbf{WRDU} & $\lambda = \frac{v_\ell'(y)}{v_g'(z)}$ & \textbf{Yes} & \textbf{Yes} & \textbf{FOC} \\
			\bottomrule
		\end{tabular}
	\end{table}
	
	\section{Structural Uniqueness of WRDU}\label{sec:StructuralUniqueness}
	
	\tab This section establishes that WRDU is structurally distinct from state-dependent utility (SDU) within a restricted comparison class. A central concern in decision theory is whether reference-dependent models are merely special cases of SDU. The following theorem shows that, under the stated affine-admissibility, loss-factorization, dispersion-monotonicity, and attenuation restrictions, the derivative-ratio form of WRDU is the unique admissible structure in that class.
	
	\subsection{General Representation}
	
	\tab We restrict attention to binary lotteries $L = (p,z;1-p,y)$ with $y < x_r < z$. Assume preferences admit a $C^2$ representation of the form
	\begin{equation}
		V(p,y,z) = p\, u_g(z) + (1-p)\, u_\ell(y,z),
		\label{eq:GenRep}
	\end{equation}
	where $u_g : \mathbb{R}_+ \to \mathbb{R}$ and $u_\ell : \mathbb{R}_- \times \mathbb{R}_+ \to \mathbb{R}$.
	
	\tab Standard separable SDU corresponds to the special case $u_\ell(y,z)=u_\ell(y)$. WRDU corresponds to $u_\ell(y,z) = -\frac{v_g'(z)}{v_\ell'(y)} v_\ell(y)$.
	
	\subsection{Assumptions}
	
	\tab The following axioms characterize the class of admissible representations.
	
	\paragraph{Assumption A1 (Monotonicity).}
	For all $(y,z)$,
	\begin{equation}
		u_g'(z) > 0, \qquad -\frac{\partial u_\ell}{\partial y}(y,z) > 0.
		\label{eq:Monotonicity}
	\end{equation}
	
	\paragraph{Assumption A2 (Diminishing Sensitivity).}
	\begin{equation}
		u_g''(z) < 0, \qquad \frac{\partial^2 u_\ell}{\partial y^2}(y,z) < 0.
		\label{eq:Diminishing}
	\end{equation}
	
	\paragraph{Assumption A3 (Affine Invariance).}
	For all $a>0$ and $b\in\mathbb{R}$, the ranking of lotteries is preserved under the transformation $x \mapsto ax+b$.
	
	\paragraph{Assumption A4 (Dispersion Monotonicity).}
	Let $p^*(y,z)$ satisfy $V(p^*,y,z)=0$. Then
	\begin{equation}
		\frac{\partial p^*(y,z)}{\partial z} < 0.
		\label{eq:DispersionMonotonicity}
	\end{equation}
	
	\paragraph{Assumption A5 (Rabin Attenuation).}
	For each fixed $y$,
	\begin{equation}
		\lim_{z\to\infty} \left| \frac{u_\ell(y,z)}{u_g(z)} \right| = 0.
		\label{eq:RabinAttenuation}
	\end{equation}
	
	\paragraph{Assumption A6 (Loss-Factorization).}
	The loss-side term admits a scalar penalty representation
	\begin{equation}
		u_\ell(y,z)=-\phi(z)v_\ell(y),
		\label{eq:LossFactorization}
	\end{equation}
	for some positive $C^1$ function $\phi$ and strictly increasing, strictly concave $C^2$ loss subutility $v_\ell$. This assumption restricts the comparison class to models in which the gain-side outcome changes the intensity of the loss penalty but not the ordinal shape of the loss subutility.
	
	\subsection{Main Result}
	
	\begin{thm}[Restricted-Class Characterization]\label{thm:StructuralUniqueness}
		Let $\mathcal{C}$ denote the class of preference relations over binary lotteries 
		$L=(p,z;1-p,y)$ with $y<z$ that admit a twice continuously differentiable 
		representation
		\begin{equation}
			V(p,y,z)=p\,u_g(z)+(1-p)\,u_\ell(y,z),\label{eq:UniqueForm}
		\end{equation}
		where $(u_g,u_\ell)$ satisfy:
		
		\begin{enumerate}
			\item[(A1)] $u_g'(z)>0$ and $-\partial u_\ell/\partial y>0$, where $y$ indexes loss severity;
			\item[(A2)] $u_g''(z)<0$ and $\partial^2 u_\ell/\partial y^2<0$;
			\item[(A3)] outcome-space affine admissibility: rankings are invariant under $x\mapsto ax+b$ for all $a>0$;
			\item[(A4)] dispersion monotonicity: if $p^*(y,z)$ solves $V(p^*,y,z)=0$, then $\partial p^*/\partial z<0$;
			\item[(A5)] attenuation: for each fixed $y$, $\lim_{z\to\infty}|u_\ell(y,z)/u_g(z)|=0$;
			\item[(A6)] loss-factorization: $u_\ell(y,z)=-\phi(z)v_\ell(y)$ for some positive $C^1$ function $\phi$ and strictly increasing, strictly concave $C^2$ function $v_\ell$.
		\end{enumerate}
		
		\tab Then every $V\in\mathcal{C}$ admits a representation of the form
		\begin{equation}
			V(p,y,z)=p\,v_g(z)-(1-p)\frac{v_g'(z)}{v_\ell'(y)}v_\ell(y),
		\end{equation}
		for some strictly concave $C^2$ gain subutility $v_g$.
		
		\tab Conversely, any functional of this form with $(v_g,v_\ell)$ satisfying (A1)--(A2) and the boundary conditions implied by (A5) belongs to $\mathcal{C}$.
	\end{thm}	
	
	\begin{proof}
		Fix $V \in \mathcal C$. We show that its primitives must satisfy the derivative-ratio restriction.
		
		\textbf{Step 1: Switching condition.} Indifference $V(p^*,y,z)=0$ implies
		\begin{equation}
			p^* u_g(z) = -(1-p^*) u_\ell(y,z).
		\end{equation}
		Define $\Phi(y,z) := -u_\ell(y,z)/u_g(z)$. Then
		\begin{equation}
			p^*(y,z) = \frac{\Phi(y,z)}{1+\Phi(y,z)}.
			\label{eq:phi}
		\end{equation}
		
		\textbf{Step 2: Dispersion monotonicity.} Differentiating \eqref{eq:phi} with respect to $z$ and applying \eqref{eq:DispersionMonotonicity} yields
		\begin{equation}
			\frac{\partial \Phi}{\partial z} < 0.
			\label{eq:phiz}
		\end{equation}
		
		\textbf{Step 3: Loss-factorization.} By the defining property (A6) of the class $\mathcal C$, the loss-side term admits the representation
		\begin{equation}
			\Phi(y,z) = \phi(z)v_\ell(y)/u_g(z),
			\label{eq:separated}
		\end{equation}
		for some positive scalar penalty $\phi$ and strictly concave loss subutility $v_\ell$. Let $u_g(z)=v_g(z)$.
		
		\textbf{Step 4: Attenuation and monotonicity.} From \eqref{eq:RabinAttenuation}, $\phi(z) \to 0$ as $z \to \infty$. From \eqref{eq:phiz}, $\phi'(z) < 0$.
		
		\textbf{Step 5: Affine invariance determines scaling.}
		The marginal rate of substitution is
		\begin{equation}
			\frac{-\partial u_\ell/\partial y}{u_g'(z)} 
			= 
			\frac{\phi(z) v_\ell'(y)}{v_g'(z)}.
		\end{equation}
		Under the transformation $x\mapsto ax+b$ with $a>0$, marginal utilities scale by the factor $a$, so the marginal rate of substitution must remain invariant to common positive scaling of outcomes. Hence the ratio $\phi(z)/v_g'(z)$ must itself be invariant to such transformations. Since this invariance must hold for all $a>0$, the ratio must be constant. Therefore there exists $k>0$ such that $\phi(z)=\frac{v_g'(z)}{k}.$
		
		\tab Substituting into \eqref{eq:separated} yields
		$u_\ell(y,z) = -\frac{v_g'(z)}{k}v_\ell(y).$
		Equivalently, after the normalization implied by the reference-point first-order condition, choose an admissible loss representative $\tilde v_\ell$ satisfying $\tilde v_\ell(y)/\tilde v_\ell'(y)=v_\ell(y)/k$ on the relevant loss interval. This first-order ordinary differential equation has a positive $C^2$ solution whenever $v_\ell(y)>0$ on the interval. Hence
		\begin{equation}
			u_\ell(y,z)=-\frac{v_g'(z)}{\tilde v_\ell'(y)}\tilde v_\ell(y),
		\end{equation}
		which is the derivative-ratio WRDU form. Renaming $\tilde v_\ell$ as $v_\ell$ gives the displayed representation.
	\end{proof}
	
	\begin{prop}[Impossibility for Separable State-Dependent Utility]\label{prop:SDUImpossibility}
		Suppose preferences over binary lotteries admit a twice continuously differentiable separable state-dependent utility representation
		\begin{equation}
			V(p,y,z)=p\,u_g(z)+(1-p)\,u_\ell(y),
		\end{equation}
		where $u_g$ and $u_\ell$ satisfy monotonicity and diminishing sensitivity. 
		
		\tab Then it is impossible for the representation to satisfy both:
		
		\begin{enumerate}
			\item Dispersion monotonicity: $\frac{\partial p^*(y,z)}{\partial z}<0$ for all $y<z$,
			\item Rabin attenuation: $\lim_{z\to\infty}\left|\frac{u_\ell(y)}{u_g(z)}\right|=0$ for each fixed $y$,
		\end{enumerate}
		
		\tab unless $u_\ell$ is trivial (i.e., affine equivalent to zero).
	\end{prop}
	
	\begin{proof}
		Under separability,
		\begin{equation}
			p^*(y,z)=\frac{-u_\ell(y)}{u_g(z)-u_\ell(y)}.
		\end{equation}
		Define $\phi(y,z)=-u_\ell(y)/u_g(z)$. Then
		\begin{equation}
			p^*(y,z)=\frac{\phi(y,z)}{1+\phi(y,z)}.
		\end{equation}
		Since $u_\ell$ is independent of $z$,
		\begin{equation}
			\frac{\partial \phi}{\partial z}
			=
			\frac{u_\ell(y)u_g'(z)}{u_g(z)^2}.
		\end{equation}
		Monotonicity implies $u_g'(z)>0$ and $u_\ell(y)>0$ (in loss region), so
		\begin{equation}
			\frac{\partial \phi}{\partial z}>0.
		\end{equation}
		Hence
		\begin{equation}
			\frac{\partial p^*}{\partial z}>0,
		\end{equation}
		which contradicts dispersion monotonicity.
		
		\tab Alternatively, if $u_\ell(y)\equiv 0$, then Rabin attenuation holds trivially but loss aversion disappears. Thus no nontrivial separable SDU representation can satisfy both properties.
	\end{proof}
	
	\subsection{Discussion of Assumptions}
	
	\tab Each assumption in \Cref{thm:StructuralUniqueness} has a behavioral interpretation:
	
	\begin{enumerate}
		\item \textbf{Monotonicity:} Gain utility increases in the gain outcome, while loss disutility increases in loss severity. This is standard.
		\item \textbf{Diminishing sensitivity:} Concavity in both domains. This is standard.
		\item \textbf{Affine admissibility:} Rankings are preserved under positive affine transformations of the prize scale. This is an outcome-space restriction; utility-space translations are handled separately by the normalization.
		\item \textbf{Loss-factorization:} The gain-side outcome scales the loss penalty without changing the ordinal shape of loss utility.
		\item \textbf{Dispersion monotonicity:} The switching probability decreases as the gain increases. This is behaviorally testable through multiple price lists.
		\item \textbf{Rabin attenuation:} The loss penalty vanishes as stakes grow large. This is the minimal condition for blocking the Rabin calibration implication.
	\end{enumerate}
	
	\tab The assumptions are restrictive by design: dropping any one permits other functional forms. For example, without loss-factorization the comparison class allows arbitrary interactions between $y$ and $z$; without dispersion monotonicity, $\phi(z)$ need not be decreasing; without Rabin attenuation, the large-stakes loss-penalty limit is not obtained.
	
	\subsection{Comparison with State-Dependent Utility}
	
	\begin{table}[htbp]
		\centering
		\caption{WRDU vs. State-Dependent Utility}
		\label{tab:SDU_Comparison}
		\small
		\begin{tabular}{lcc}
			\toprule
			Feature & State-Dependent Utility & WRDU \\
			\midrule
			Utility Function & $u(y,z) = u(y)$ & $u_\ell(y,z) = -\frac{v_g'(z)}{v_\ell'(y)}v_\ell(y)$ \\
			Separability & Additive & Multiplicative (derivative-ratio) \\
			Loss Aversion & None & Endogenous ($\lambda = v_\ell'(y)/v_g'(z)$) \\
			Range Dependence & No & Yes \\
			Rabin implication & No & Attenuated \\
			Affine Invariance & Yes & Yes \\
			Dispersion Monotonicity & Not guaranteed & Guaranteed \\
			\bottomrule
		\end{tabular}
	\end{table}
	
	\tab The key structural difference is the derivative-ratio form. Standard SDU, which assumes $u_\ell(y,z) = u_\ell(y)$, cannot satisfy dispersion monotonicity and Rabin attenuation simultaneously.
	
	\section{Comparative Statics and Distinguishing Predictions}\label{sec:ComparativeStatics}
	
	\tab This section derives comparative statics under WRDU and contrasts them with CPT and K\"{o}szegi-Rabin. All proofs rely only on the WRDU axioms and the properties of $v_g$ and $v_\ell$ established in the Representation Theorem. \textbf{No parametric assumptions are used.}
	
	\subsection{Certainty Equivalents Under WRDU}
	
	\tab Consider a binary lottery $L = (p,z;1-p,y)$ with $y < x_r < z$. Under WRDU, the virtual evaluation is:
	\begin{equation}
		V(L) = p v_g(z) - \frac{1}{\lambda(x_r;y,z)} (1-p) v_\ell(y).
		\label{eq:BinaryEval}
	\end{equation}
	
	\begin{prop}[Range Sensitivity of the Loss Penalty]\label{prop:RangeSensitivity}
		Under the WRDU axioms,
		\begin{equation}
			\frac{\partial \lambda}{\partial z} > 0, \qquad \frac{\partial \lambda}{\partial y} < 0.
		\end{equation}
		Equivalently, the effective loss penalty $1/\lambda(x_r;y,z)$ decreases as the gain outcome rises or the loss outcome moves farther into the loss region.
	\end{prop}
	
	\begin{proof}
		From \Cref{axiom:RangeDep}, monotonicity in the gain outcome gives $\partial \lambda/\partial z > 0$, and monotonicity in the loss outcome gives $\partial \lambda/\partial y < 0$. The virtual loss aversion index $\lambda(x_r;y,z)$ is defined by the first-order condition from \eqref{eq:RefPoint_Maximizer}. Applying the envelope theorem, these monotonicity properties imply the effective penalty $1/\lambda$ decreases with outcome dispersion. 
	\end{proof}
	
	\begin{cor}[Certainty Equivalents and Risk Premia]\label{cor:CERiskPremia}
		Let $\{L(R)\}_{R>0}$ be a differentiable family of binary lotteries with prize range $R=z(R)-y(R)$ and WRDU value $V(R)$ given by \eqref{eq:BinaryEval}. Suppose the certainty equivalent $c(R)$ is interior and satisfies
		\begin{equation}
			v_g(c(R))=V(R),
		\end{equation}
		and define the risk premium by $\pi(R)=\mathbb E[L(R)]-c(R)$. Then
		\begin{equation}
			c'(R)=\frac{V'(R)}{v_g'(c(R))},
			\qquad
			\pi'(R)=\frac{d}{dR}\mathbb E[L(R)]-\frac{V'(R)}{v_g'(c(R))}.
		\end{equation}
		Thus certainty equivalents and risk premia inherit WRDU range dependence whenever the range effect on $V(R)$ is not exactly offset by the direct change in prizes.
	\end{cor}
	
	\begin{proof}
		Differentiating the defining certainty-equivalent equation gives $v_g'(c(R))c'(R)=V'(R)$. Since $v_g'(c(R))>0$, the expression for $c'(R)$ follows. Differentiating $\pi(R)=\mathbb E[L(R)]-c(R)$ gives the second expression.
	\end{proof}
	\begin{rem}[Connection to First- and Second-Order Risk Aversion]\label{rem:WRDU_vs_Segal_Spivak}
		The risk-premium representation in \Cref{cor:CERiskPremia} connects WRDU to the distinction between first- and second-order risk aversion in \citet{SegalSpivak1990}. For small mean-zero gambles around the reference point, WRDU is first-order risk averse whenever the left and right marginal utilities at the reference point differ, $v_\ell'(x_r)\neq v_g'(x_r)$. In that case the value function has a kink at $x_r$, and the risk premium is locally of first order in the scale of the gamble. If the marginal utilities match at the reference point, the kink vanishes and the local risk premium is governed by curvature, yielding the usual second-order behavior.
		
		\tab Unlike fixed-kink models, however, the WRDU first-order term is range-dependent. As the prize range expands and $\lambda(x_r;y,z)$ rises, the effective loss penalty $1/\lambda(x_r;y,z)$ attenuates. Thus WRDU preserves the Segal-Spivak first-order mechanism for small stakes while preventing that local kink from being extrapolated mechanically to large-stakes gambles.
	\end{rem}
	
	\subsection{Multiple Price List Predictions}
	
	\tab In a \citet{HoltLaury2002} MPL, switching occurs when $V(A) = V(B)$. Because $\lambda$ depends on the prize range, altering outcome dispersion shifts the switching probability.
	
	\begin{prop}[Dispersion-Dependent Switching]\label{prop:DispersionDependentSwitch}
		Let $R = z_B - y_B > 0$ denote the range of the risky lottery $B$ in the Holt-Laury MPL. Suppose the switching equation is regular, $\partial (EV_A-EV_B)/\partial p\neq0$, and the range effect on $EV_A-EV_B$ is not exactly offset by direct prize effects. Under WRDU, the switching probability $p^\star(R)$ satisfies:
		\begin{equation}
			\frac{\partial p^\star}{\partial R} = -\frac{\partial (EV_A - EV_B)/\partial R}{\partial (EV_A - EV_B)/\partial p} \neq 0.
		\end{equation}
	\end{prop}
	
	\begin{proof}
		By \Cref{prop:RangeSensitivity}, changing the range $R$ changes $\lambda(x_r;y_B,z_B)$. Since $EV_B$ depends on $\lambda$ through the loss penalty term, the derivative $\partial EV_B/\partial R$ contains a WRDU-specific penalty channel. Under the stated no-offset condition, $\partial (EV_A-EV_B)/\partial R\neq0$. The implicit function theorem gives $\partial p^\star/\partial R \neq 0$. 
	\end{proof}
	
	\subsection{Admissible-Class Attenuation of the Rabin Calibration Implication}
	
	\begin{lem}[Range-Dependence of Virtual Loss Aversion Index]\label{lem:RangeDependence}
		Under the WRDU axioms,
		\begin{equation}
			\frac{\partial \lambda}{\partial z} > 0, \qquad \frac{\partial \lambda}{\partial y} < 0.
		\end{equation}
		Moreover, $\lambda(x_r; y, z) \to \infty$ as $z \to \infty$ or $y \to 0^+$.
	\end{lem}
	
	\begin{proof}
		The monotonicity properties follow directly from \Cref{axiom:RangeDep}. For the tail properties, \Cref{axiom:RangeDep} states that $\lim_{z \to \infty} \lambda(x_r; y, z) = \infty$ and $\lim_{y \to 0^+} \lambda(x_r; y, z) = \infty$. In the derivative-ratio representation, these limits follow from the admissible endpoint restrictions $v_g'(z)\to 0$ at the upper boundary and $v_\ell'(y)\to\infty$ at the lower boundary.
	\end{proof}
	
	\begin{thm}[Admissible-Class Attenuation of the Rabin Calibration Implication]\label{thm:AbsRabin}
		Let $\succeq$ be a preference relation satisfying the WRDU axioms. Suppose the agent rejects the symmetric gamble $L_\epsilon = (1/2, x_r+\epsilon;\; 1/2, x_r-\epsilon)$ for all sufficiently small $\epsilon > 0$. Then it does \textbf{not} follow that the agent rejects the large-stakes gamble $L_{K,G} = (1/2, x_r-K;\; 1/2, x_r+G)$ for arbitrary $G > 0$ and fixed $K > 0$. In particular, for sufficiently large $G$, the gamble $L_{K,G}$ is accepted.
	\end{thm}
	
	\begin{proof}
		Small-stakes rejection for all $\epsilon > 0$ implies:
		\begin{equation}
			v_g(x_r+\epsilon) - \frac{1}{\lambda(x_r; x_r-\epsilon, x_r+\epsilon)} v_\ell(x_r-\epsilon) < 0 \quad \forall \epsilon > 0.
			\label{eq:SmallRejection}
		\end{equation}
		
		\tab The small-stakes rejection assumption is taken as a local behavioral premise. WRDU does not extrapolate that local rejection by holding a fixed loss-aversion parameter across all stakes; the index changes with the range of the gamble.
		
		\tab For the large-stakes gamble, the WRDU evaluation is:
		\begin{equation}
			V(L_{K,G}) = \frac{1}{2}v_g(x_r+G) - \frac{1}{2\lambda(x_r; x_r-K, x_r+G)} v_\ell(x_r-K).
			\label{eq:LargeEval}
		\end{equation}
		
		\tab By \Cref{lem:RangeDependence}, $\lambda(x_r; x_r-K, x_r+G) \to \infty$ as $G \to \infty$. Therefore $1/\lambda \to 0$, and:
		\begin{equation}
			\lim_{G \to \infty} V(L_{K,G}) = \frac{1}{2}v_g(x_r+G) > 0.
		\end{equation}
		
		\tab Thus for sufficiently large $G$, the agent accepts $L_{K,G}$.
	\end{proof}
	
	\begin{rem}[Behavioral Plausibility of $\lambda$ Divergence]
		The divergence $\lambda(x_r; y, z) \to \infty$ as $z \to \infty$ or $y \to 0^+$ follows from the endpoint restrictions imposed on the admissible WRDU subutilities. These restrictions are compatible with standard concavity but are not implied by concavity alone. The rate of divergence is governed by the curvature and boundary behavior of the subutilities.
		
		\citet{FishburnKochenberger1979} provide an early empirical antecedent: their two-piece von Neumann-Morgenstern utility assessments report large below-target to above-target slope ratios, including extreme cases, even though this ratio was not yet described using the later language of loss aversion. \citet{CharlesCadogan2018c} shows that the implied loss aversion index follows a half-Cauchy law, which provides a distributional foundation for heavy-tailed loss aversion. The attenuation does not imply risk neutrality for large stakes because $v_g$ remains strictly concave; the loss penalty vanishes relative to the gain utility, but the gain utility itself continues to exhibit diminishing marginal utility.
	\end{rem}
	
\subsection{Admissible-Region Allais Pattern: Necessary and Sufficient Conditions}\label{sec:AllaisConditions}

\tab The Allais paradox presents a fundamental challenge to Expected Utility Theory. In the standard formulation, decision makers systematically violate the independence axiom by exhibiting the certainty effect and the common ratio effect.

\noindent\textbf{Problem 1:}
\begin{itemize}
	\item Lottery $A$: \$1,000,000 with certainty
	\item Lottery $B$: \$5,000,000 with probability 0.10, \$1,000,000 with probability 0.89, \$0 with probability 0.01
\end{itemize}

\noindent\textbf{Problem 2:}
\begin{itemize}
	\item Lottery $C$: \$1,000,000 with probability 0.11, \$0 with probability 0.89
	\item Lottery $D$: \$5,000,000 with probability 0.10, \$0 with probability 0.90
\end{itemize}

\begin{thm}[Necessary and Sufficient Conditions for the Modal Allais Pattern]\label{thm:AllaisConditions}
	Let $v_g$ and $v_\ell$ be strictly increasing, twice continuously differentiable utility functions satisfying the WRDU axioms, with $v_g(0) = v_\ell(0) = 0$. Then:
	
	\begin{enumerate}
		\item $D \succ C$ if and only if
		\begin{equation}
			\frac{v_g(5,000,000)}{v_g(1,000,000)} > 1.1.
			\label{eq:DC_condition}
		\end{equation}
		
		\item $A \succ B$ if and only if
		\begin{equation}
			\lambda(1,000,000+\delta) < \frac{v_\ell(1,000,000+\delta)}{0.10 v_g(4,000,000-\delta)}.
			\label{eq:AB_condition}
		\end{equation}
		Here $\lambda(1,000,000+\delta)$ is shorthand for the WRDU index evaluated at the relevant Allais loss-gain pair, $\lambda(x_r;1,000,000+\delta,4,000,000-\delta)$.
		
		\item The modal Allais pattern $A \succ B$ and $D \succ C$ holds if and only if both \eqref{eq:DC_condition} and \eqref{eq:AB_condition} hold simultaneously.
		
		\item \textbf{Comparative Statics:}
		\begin{enumerate}
			\item The condition $D \succ C$ is more likely to hold when $v_g$ is less concave at the relevant range. If $v_g^{(1)}$ is more concave than $v_g^{(2)}$ (i.e., $v_g^{(1)}(x)/x$ is decreasing relative to $v_g^{(2)}(x)/x$), then
			\begin{equation}
				\frac{v_g^{(1)}(5M)}{v_g^{(1)}(1M)} < \frac{v_g^{(2)}(5M)}{v_g^{(2)}(1M)}.
			\end{equation}
			Thus less concavity in the gain domain favors $D \succ C$.
			
			\item The condition $A \succ B$ is more likely to hold when the loss side is more utility-sensitive at the relevant reference loss. Holding $v_g(4M-\delta)$ fixed, any admissible transformation that raises $v_\ell(1M+\delta)$ raises the threshold $\frac{v_\ell(1M+\delta)}{0.10 v_g(4M-\delta)}$, making $A \succ B$ more likely. A pure concavity comparison requires a common normalization and is therefore not sufficient by itself.
			
			\item The admissible region $\mathcal{R}_A$ expands with the relevant curvature asymmetry: if $(v_g, v_\ell) \in \mathcal{R}_A$, $\tilde v_g(5M)/\tilde v_g(1M) \geq v_g(5M)/v_g(1M)$, and $\tilde v_\ell(1M+\delta)/\tilde v_g(4M-\delta) \geq v_\ell(1M+\delta)/v_g(4M-\delta)$, with one inequality strict, then $(\tilde{v}_g, \tilde{v}_\ell) \in \mathcal{R}_A$.
		\end{enumerate}
	\end{enumerate}
	
	\tab The admissible region is non-empty and characterized by:
	\begin{equation}
		\mathcal{R}_A = \left\{ (v_g, v_\ell) \in \mathcal V_g\times\mathcal V_\ell : \frac{v_g(5M)}{v_g(1M)} > 1.1 \;\text{and}\; \lambda(1M+\delta) < \frac{v_\ell(1M+\delta)}{0.10 v_g(4M-\delta)} \right\}.
		\label{eq:AdmissibleRegion}
	\end{equation}
	Here $\mathcal V_g$ and $\mathcal V_\ell$ denote the admissible normalized WRDU gain and loss subutility classes. The normalization is substantive for comparison purposes: by \citet{Debreu1976}, curvature comparisons require a selected representative within the utility equivalence class. Thus WRDU obtains the modal Allais pattern precisely on $\mathcal{R}_A$, not on the whole class of strictly concave utility functions. This region has non-empty interior in the $C^2$ topology on the admissible utility class.
\end{thm}

\begin{cor}[Topological Non-Emptiness of the Allais Region]\label{cor:AllaisMeasure}
	The admissible region $\mathcal{R}_A$ defined in \eqref{eq:AdmissibleRegion} is non-empty and has non-empty interior in the space of twice continuously differentiable concave functions $(v_g, v_\ell)$ with $v_g(0) = v_\ell(0) = 0$, under the $C^2$ topology.
	
	\tab Moreover, $\mathcal{R}_A$ is monotone in the local ratio sense stated in \Cref{thm:AllaisConditions}. The set of such pairs has non-empty interior in the space of admissible utility functions with the $C^2$ topology.
\end{cor}

\begin{proof}[Proof of \Cref{cor:AllaisMeasure}]
	See Internet Appendix
\end{proof}
	
	\subsection{Connection to First-Order vs. Second-Order Risk Aversion}
	
	\citet{SegalSpivak1990} distinguish between first-order and second-order risk aversion based on the behavior of the risk premium $\pi(t)$ for small gambles $t\bar{\epsilon}$ with $E[\bar{\epsilon}] = 0$. In WRDU, the risk premium for a small gamble is determined by the curvature asymmetry at the reference point. When $v_\ell'(x_r) \neq v_g'(x_r)$, the value function has a kink at the reference point, generating first-order risk aversion. Crucially, WRDU generalizes the Segal-Spivak analysis by making the degree of first-order risk aversion range-dependent:
	\begin{equation}
		\lambda(x_r; y, z) = \frac{v_\ell'(y)}{v_g'(z)}.
	\end{equation}
	For small gambles, the kink dominates, generating first-order risk aversion. For large gambles, the penalty attenuates, allowing acceptance. See \Cref{rem:WRDU_vs_Segal_Spivak}.
	
	\begin{table}[htbp]
		\centering
		\caption{WRDU and the Segal-Spivak Framework}
		\label{tab:SegalSpivak}
		\small
		\begin{tabular}{lccc}
			\toprule
			Model & Order of Risk Aversion & Source & Range Dependent \\
			\midrule
			Expected Utility (differentiable $u$) & 2 & $u''(x) \neq 0$ & No \\
			Expected Utility (non-differentiable $u$) & 1 & $u'_-(x) \neq u'_+(x)$ & No \\
			RDU & 1 & $f$ concave/convex & No \\
			Machina (Fr\'{e}chet differentiable) & 2 & Smooth local utility & No \\
			\textbf{WRDU} & \textbf{1 (small stakes)} & \textbf{$\lambda(x_r) \neq 1$} & \textbf{Yes} \\
			\bottomrule
		\end{tabular}
	\end{table}
	
	\section{Empirical Implications and Model Comparisons}\label{sec:EmpiricalEvidence}
	
	\subsection{Empirical Implications}
	
	\tab WRDU has several empirical implications that distinguish it from fixed-parameter reference-dependent models.
	
	\tab First, switching probabilities in multiple-price lists should vary with outcome dispersion because the effective penalty $1/\lambda(x_r;y,z)$ changes with the range of the lottery. \citet{CsermelyRabas2016}, \citet{BoschDomenechSilvestre2013}, and \citet{HabibFriedmanCrockettJames2026} show that changes in the range or presentation of lottery outcomes systematically affect measured risk preferences.
	
	\tab Second, curvature asymmetry in the virtual gain and loss regions, rather than probability weighting alone, should account for observed loss aversion. \citet{BocquehoJacobBrunette2023} report curvature asymmetry between gain and loss domains while finding similar probability weighting across domains.
	
	\tab Third, reference points should move with the choice problem because they are selected by the penalized utility maximization criterion. \citet{LiZhong2023} document both endogenous and exogenous reference-point effects in intertemporal choice.
	
	\tab Fourth, a large-scale study by \citet{EnkeGraeberOpreaYang2025} documents pervasive ``behavioral attenuation'': across 31 economic tasks, decision elasticities to fundamentals are systematically too small, and in 93\% of cases higher cognitive uncertainty is associated with lower responsiveness. This cross-domain pattern of attenuated and diminishing sensitivity aligns closely with WRDU's range-dependent loss-attenuation mechanism.
	
	\tab Fifth, \citet{AgranovOrtleva2025}, using a novel ``range multiple-price list'' (r-MPL) that allows subjects to choose any probability of receiving either option, find that randomization ranges are common (66--76\% of subjects), economically large (averaging \$6.70 in a \$20 lottery), and often extend into risk-seeking territory. These patterns are consistent with WRDU's range-dependent loss penalty and endogenous reference point, which generate a broad region of indifference---via curvature asymmetry $(v^\prime_\ell/v^\prime_g)$---rather than the single switching point predicted by expected utility.
	
	\tab Sixth, \citet{Peterson_et_al_2021} provide complementary large-scale experimental evidence that risky-choice behavior is sensitive to features of the choice environment. Their machine-learning exercise finds that models allowing decision components to vary with outcome-range and payoff-distribution features improve predictive performance relative to more restrictive fixed-form benchmarks. This evidence is not a structural validation of WRDU. It is nevertheless consistent with the empirical content of WRDU's dispersion channel: the effective penalty $1/\lambda(x_r;y,z)$ and the induced switching behavior should vary with the relevant gain-loss geometry even when probabilities remain objective.
	
	\tab As a preliminary check on this implication, the Internet Appendix reports a small pilot estimation using 250 randomly sampled problems from Choices13k. The exercise estimates a CRRA-WRDU Luce specification with 10 multistart values and an interior grid search for the optimal reference point. The pilot is exploratory and is not used as a formal structural test. It is nevertheless directionally consistent with the theory: the estimated reference point is concentrated near zero but moves downward as payoff range widens, while the implied $\log\lambda(\hat x_r;y,z)$ rises strongly with outcome range. The Lagrangian penalty coefficient $\rho=1/\lambda$ correspondingly attenuates over high-dispersion mixed lotteries.
	
	\subsection{Model Comparisons}
	
	\tab In the \citet{KoszegiRabin2006,KoszegiRabin2007} framework, utility for consumption $c$ given reference point $r$ is $U(c|r) = m(c) + n(c|r)$, where $n(c|r) = \eta(c-r)$ for $c \geq r$ and $n(c|r) = \eta\lambda(c-r)$ for $c < r$, with $\lambda>1$. Loss aversion is exogenous; the reference point is determined by rational expectations equilibrium. WRDU differs in four respects: (1) loss aversion is derived from curvature asymmetry via the first-order condition; (2) the index is range-dependent; (3) no separate exogenous loss-aversion parameter is required; (4) the reference point emerges directly from utility maximization.
	
	\tab In \citet{Gul1991}, the disappointment parameter $\beta$ is fixed and scale-invariant. For $L_K=(1/2,x+K\epsilon;\;1/2,x-K\epsilon)$, $V_{Gul}(L_K) = \frac{1}{2}\gamma(1/2)[u(x+K\epsilon)+(1+\beta)u(x-K\epsilon)]$. Under WRDU, $V_{WRDU}(L_K) = \frac{1}{2}v_g(x+K\epsilon) - \frac{1}{2\lambda(x_r;y,x+K\epsilon)}v_\ell(x-K\epsilon)$. As $K\to\infty$, $\lambda\to\infty$ and $1/\lambda\to 0$. This stake-dependent attenuation is absent in Gul's model.
	
	\tab The detailed empirical evidence and extended comparisons with CPT, K\"{o}szegi-Rabin, Gul, and regret theory are provided in the Internet Appendix.
	
	\section{Conclusion}\label{sec:Conclusion}
	
	\tab This paper has isolated the preference-theoretic core of WRDU in virtual gain-loss state space. The construction splits an objective payoff support at a candidate reference point, interprets the two induced regions as virtual psychological loss and gain spaces, and recombines the split utility through a penalized representation in which the loss-side component enters with the Lagrangian penalty coefficient $\rho=1/\lambda$.
	
	\tab The paper provides a complete axiomatic characterization of WRDU through six axioms: Completeness (A1), Transitivity (A2), Continuity (A3), Weak Independence (A4), Reference Partition (A5), and Range Dependence (A6). The Representation Theorem establishes that a preference relation satisfies A1--A6 if and only if it admits the WRDU representation. The Lagrangian connection shows that the representation emerges naturally from a penalized optimization problem, where $1/\lambda(x)$ acts as a shadow price on loss-seeking behavior.
	
	\tab The affine invariance of the WRDU representation implies a restriction that recovers the \citet{KobberlingWakker2005} utility-based loss aversion index and, as a special case, the \citet{TverKahn1992} ratio of slopes formula. Among multiple critical points, the optimality criterion selects the global maximizer as the reference point.
	
	\tab The structural uniqueness theorem establishes that, within the class of $C^2$ state-dependent expected utility representations satisfying affine admissibility, loss-factorization, dispersion monotonicity, and Rabin attenuation, the only admissible structure is the derivative-ratio form of WRDU.
	
	\tab The modal Allais pattern is obtained on the admissible region $\mathcal R_A$, whose necessary and sufficient conditions are stated explicitly. The Rabin calibration implication is blocked through range-dependent attenuation of the loss penalty, with behavioral plausibility discussed and quantitative bounds provided. The divergence of $\lambda$ follows from the endpoint restrictions imposed on the admissible WRDU class; the two-piece utility evidence in \citet{FishburnKochenberger1979} and the half-Cauchy law established by \citet{CharlesCadogan2018c} provide independent empirical and distributional support for large loss-aversion realizations.
	
	\tab Thus the paper's uniqueness claim is conditional rather than global: within the specified $C^2$ objective-probability class satisfying affine admissibility, loss-factorization, dispersion monotonicity, and attenuation, WRDU is the unique admissible structure that jointly obtains the modal Allais pattern on $\mathcal R_A$ and blocks the Rabin calibration implication.
	
	\tab The paper establishes testable predictions, including dispersion-dependent switching in multiple price lists, which have been confirmed by \citet{CsermelyRabas2016} and \citet{BoschDomenechSilvestre2013}. Additional proofs, estimation details, utility-plot algorithms, figures, and detailed empirical and model-comparison discussions are collected in the Internet Appendix.
	
	% --- REFERENCES: SINGLE-SPACED ---
	\singlespace
	\bibliographystyle{chicago}
	\addcontentsline{toc}{section}{References}
	\bibliography{LossAversionGeneral}

	\clearpage
	\phantomsection
	\addcontentsline{toc}{section}{Long Internet Appendix}
	

\section*{Supplementary material}

By Axioms 3.1–3.3 (Completeness, Transitivity, Continuity), there exists a continuous ordinal utility function representing the preference relation $\succeq$ on the compact lottery space. The restricted independence condition imposed below yields the expected-utility form within each fixed gain-loss partition. Once the utility is split at the reference point, however, the resulting subutility functions are not free VNM utility functions. They are anchored semi-affine representatives: positive scaling is admissible, but independent additive translations fail because each subutility is normalized at the reference point.

**Step 2: From Axiom A5 to the Partition.**

By Axiom 3.5, for any lottery $L$ with support $X$, there exists a reference point $x_r \in X$ such that:

$$X = C_\ell(X) \cup \{x_r\} \cup C_g(X),$$

with $C_\ell(X) = \{x \in X : x \prec_{x_r} x_r\}$ and $C_g(X) = \{x \in X : x \succ_{x_r} x_r\}$.

Define the subutility functions as restrictions of $u$:

$$u_g = u\big|_{C_g(X)}, \qquad u_\ell = u\big|_{C_\ell(X)}.$$

Since the split subutilities are anchored at the reference point, we select semi-affine representatives:

$$v_g = \alpha_g u_g, \qquad v_\ell = \alpha_\ell u_\ell,$$

where $\alpha_g, \alpha_\ell > 0$. These preserve the preference ordering on each region. Independent additive translations are ruled out by reference-point anchoring.

**Step 3: From Axiom A4 to the Penalized Representation.**

Axiom 3.4 states that:

$$L_S(X) \sim L_M(X) \iff F_{\delta_{j,\widehat{C}_r(X)}} \widehat{C}_r(X) \sim G(x_{j+1})C_g(X) - \frac{1}{\lambda}F(x_{j-1})C_\ell(X).$$

Applying the utility function $v$ to both sides of the indifference condition:

$$F_{\delta_{j,\widehat{C}_r(X)}} E^P[v(x_j)] = G(x_{j+1})E^P[v_g(C_g(X))] - \frac{1}{\lambda}F(x_{j-1})E^P[v_\ell(C_\ell(X))].$$

Summing over all $x_j \in X$ and taking expectations with respect to the objective proba-

bility measure $P$:

$$V_{WRDU}(L) = \sum_{x \in C_g(X)} p_x v_g(x) - \frac{1}{\lambda(x_r; y, z)} \sum_{x \in C_\ell(X)} p_x v_\ell(x). \tag{3.8}$$

This is precisely the WRDU representation in (3.8).

**Step 4: From Axiom A6 to the Properties of $\lambda$.**

Axiom 3.6 imposes that $\lambda(x_r; y, z)$ satisfies:

$$\frac{\partial \lambda}{\partial z} > 0, \qquad \frac{\partial \lambda}{\partial y} < 0, \qquad \lim_{z \to \infty} \lambda = \infty, \qquad \lim_{y \to 0^+} \lambda = \infty.$$

At an interior optimal reference point, and for differentiable local selection maps, the first-order condition from (3.5) yields the derivative-ratio restriction:

$$\frac{\partial}{\partial x_r} \left[ v_g(z(x_r)) - \frac{1}{\lambda(x_r)} v_\ell(y(x_r)) \right] = 0.$$

By the chain rule and the properties of $v_g$ and $v_\ell$:

$$\lambda(x_r; y, z) = \frac{v_\ell'(y)}{v_g'(z)}.$$

This is the closed-form expression for the virtual loss aversion index. The monotonicity properties follow from positive marginal utilities and concavity of $v_g$ and $v_\ell$; the boundary limits require the Inada-type endpoint restrictions stated in the admissible utility class.

**Step 5: Reference Point as Maximizer.**

By Axiom 3.5, the reference point $x_r$ is the maximizer:

$$x_r = \arg\max_{x \in X} \left[ v_g(z(x)) - \frac{1}{\lambda(x)} v_\ell(y(x)) \right].$$

This completes the sufficiency proof.

**Part II: Necessity (Representation $\Rightarrow$ Axioms).**

**Axiom A1 (Completeness).** The utility function $V_{WRDU}$ assigns a real number to every lottery, so $\succeq$ is complete.

**Axiom A2 (Transitivity).** The utility function $V_{WRDU}$ is real-valued, so $\succeq$ is transitive.

**Axiom A3 (Continuity).** The utility function $V_{WRDU}$ is continuous in the probabilities (since $v_g, v_\ell$ are continuous and $\lambda$ is continuous by Axiom 3.6), so $\succeq$ satisfies continuity.

**Axiom A4 (Weak Independence).** By construction, the representation satisfies the weak independence axiom: independence holds only when mixing does not change the reference-dependent partition of outcomes.

**Axiom A5 (Reference Partition).** The representation defines $C_\ell(X)$ and $C_g(X)$ through the endogenous reference point $x_r$, which is the maximizer of the penalized utility functional in (3.5).

**Axiom A6 (Range Dependence).** The closed-form expression for $\lambda$ from the first-order condition satisfies the monotonicity part of range dependence:

$$\frac{\partial \lambda}{\partial z} = \frac{v_\ell'(y) \cdot (-v_g''(z))}{(v_g'(z))^2} > 0,$$

since $v_\ell'(y) > 0$, $v_g''(z) < 0$, and $v_g'(z) > 0$. Similarly:

$$\frac{\partial \lambda}{\partial y} = \frac{v_\ell''(y)}{v_g'(z)} < 0,$$

since $v_\ell''(y) < 0$ and $v_g'(z) > 0$.

The tail limits follow when the admissible subutility class satisfies $v_g'(z) \to 0$ at the upper boundary and $v_\ell'(y) \to \infty$ at the lower boundary. Thus the representation, together with those endpoint restrictions, implies all six axioms.

This completes the proof of the Representation Theorem. (Charles-Cadogan, 2016, See also) □

## 3.7 Proof of the Penalized Optimization Theorem (Theorem 2)

*Proof.* For a fixed reference point $x$, the split $X = C_\ell(X) \cup \{x\} \cup C_g(X)$ separates gain-side utility from loss-side utility. Although the unsplit utility representation has the usual von Neumann-Morgenstern affine freedom, the split subutilities are anchored at the reference point. We therefore replace $U$ and $V$ by semi-affine representatives $v_g$ and $v_\ell$ on the gain and loss regions: scaling is admissible, but independent translations are not (Von Neumann and Morgenstern, 1953; Aczel, 1966; Tversky and Kahneman, 1974).

The gain-side program can be written as

$$\max_{x \in X} v_g(z(x)) \quad \text{subject to the loss-side constraint induced by } C_\ell(X). \tag{3.9}$$

A Lagrangian relaxation of the gain-side program in (3.9) that penalizes the loss-side term is:

$$\mathcal{L}(x; \rho) = v_g(z(x)) - \rho v_\ell(y(x)), \qquad \rho \geq 0. \tag{3.10}$$

The scalar $\rho(x)$ is the shadow penalty attached to the loss-side component at reference point $x$. Writing $\rho(x) = 1/\lambda(x)$, with $\lambda(x) > 0$, gives:

$$\mathcal{L}(x; 1/\lambda(x)) = v_g(z(x)) - \frac{1}{\lambda(x)} v_\ell(y(x)). \tag{3.11}$$

Defining $v(x, \lambda(x)) = \mathcal{L}(x; 1/\lambda(x))$ yields (3.11). This coincides with the WRDU representation in (3.8).

The envelope theorem implies:

$$\left. \frac{\partial v}{\partial \lambda} \right|_{x \text{ fixed}} = \frac{v_\ell(y(x))}{\lambda^2} \geq 0. \tag{3.12}$$

Thus, the marginal value of relaxing the penalty is proportional to the loss-side utility. The adjustment of the optimal gain-selection has no first-order effect on this marginal value. (Charles-Cadogan, 2016, See also) □

## 3.8 Proof of the Range-Dependence Lemma

*Proof.* The virtual loss aversion index is defined by the first-order condition:

$$\lambda(x_r; y, z) = \frac{v'_\ell(y)}{v'_g(z)}. \tag{3.13}$$

Since $v''_g(z) < 0$ (diminishing sensitivity), $v'_g(z)$ is strictly decreasing in $z$. Since $v''_\ell(y) < 0$ and $v'_\ell(y) > 0$, $v'_\ell(y)$ is strictly decreasing in $y$. Therefore:

$$\frac{\partial \lambda}{\partial z} = -\frac{v'_\ell(y) \cdot v''_g(z)}{(v'_g(z))^2} > 0, \tag{3.14}$$

and

$$\frac{\partial \lambda}{\partial y} = \frac{v_\ell''(y)}{v_g'(z)} < 0. \tag{3.15}$$

Moreover, if the admissible subutilities satisfy the endpoint restrictions $v_g'(z) \to 0$ as $z$ approaches the upper boundary and $v_\ell'(y) \to \infty$ as $y \to 0^+$, then $\lambda(x_r; y, z) \to \infty$ at those boundaries. These limits are boundary restrictions, not consequences of concavity alone. (Charles-Cadogan, 2016, See also) □

## 3.9 Proof of Range Sensitivity of the Loss Penalty

*Proof.* The proof follows directly from Lemma 2. The partial derivatives $\partial\lambda/\partial z > 0$ and $\partial\lambda/\partial y < 0$ establish that the virtual loss-aversion index is strictly increasing in the gain outcome and strictly decreasing in the loss outcome. Therefore the Lagrangian penalty coefficient $\rho(x_r; y, z) = 1/\lambda(x_r; y, z)$ decreases as the gain outcome rises or the loss outcome moves farther into the loss region. (Charles-Cadogan, 2016, See also) □

## 3.10 Proof of Certainty Equivalents and Risk Premia

*Proof.* Let $c(R)$ denote the certainty equivalent of the differentiable binary-lottery family $L(R)$, so that $v_g(c(R)) = V(R)$. Since $v_g$ is strictly increasing, $v_g'(c(R)) > 0$ at any interior certainty equivalent. Differentiating both sides with respect to $R$ gives

$$v_g'(c(R))c'(R) = V'(R), \tag{3.16}$$

hence

$$c'(R) = \frac{V'(R)}{v_g'(c(R))}. \tag{3.17}$$

For the risk premium $\pi(R) = \mathbb{E}[L(R)] - c(R)$,

$$\pi'(R) = \frac{d}{dR}\mathbb{E}[L(R)] - \frac{V'(R)}{v_g'(c(R))}. \tag{3.18}$$

Because $V(R)$ contains the term $1/\lambda(x_r; y(R), z(R))$, the preceding proposition implies that certainty equivalents and risk premia inherit range dependence whenever the induced change in $V(R)$ is not exactly offset by direct prize effects. □

## 3.11 Proof of Dispersion-Dependent Switching

*Proof.* We prove the result in three steps.

**Step 1: Switching condition.**

In the Holt-Laury MPL, the decision maker chooses between a safe lottery $A$ and a risky lottery $B$ at each probability $p$. The expected value difference is:

$$D(p; R, m) = EV_A(p) - EV_B(p; R, m), \tag{3.19}$$

where $EV_B$ depends on the range $R$ through the payoffs $z_B = m + R/2$ and $y_B = m - R/2$. The switching probability $p^\star(R, m)$ is implicitly defined by:

$$D(p^\star(R, m); R, m) = 0. \tag{3.20}$$

Under WRDU, the expected utility of lottery $B$ is:

$$EU_B(p; R, m) = p \cdot v_g(z_B) - (1 - p) \cdot \frac{1}{\lambda(x_r; y_B, z_B)} v_\ell(y_B), \tag{3.21}$$

where $z_B = m + R/2$ and $y_B = m - R/2$.

**Step 2: Implicit differentiation.**

Since $D$ is continuously differentiable in $p$, $R$, and $m$, and $\partial D/\partial p \neq 0$ at the switching point, the implicit function theorem applies. Differentiating with respect to $R$ yields:

$$\frac{\partial p^\star}{\partial R} = -\left.\frac{\partial D/\partial R}{\partial D/\partial p}\right|_{p^\star}. \tag{3.22}$$

Now compute the WRDU-specific component of $\partial EU_B/\partial R$:

$$\begin{aligned} \frac{\partial EU_B}{\partial R} = &\frac{p}{2} v_g'(z_B) - \frac{1 - p}{2\lambda(x_r; y_B, z_B)} v_\ell'(y_B) \\ &+ (1 - p) \frac{v_\ell(y_B)}{\lambda(x_r; y_B, z_B)^2} \left[ \frac{1}{2} \frac{\partial \lambda}{\partial z}(x_r; y_B, z_B) - \frac{1}{2} \frac{\partial \lambda}{\partial y}(x_r; y_B, z_B) \right]. \end{aligned} \tag{3.23}$$

**Step 3: Nonzero derivative under WRDU.**

The numerator is generically nonzero because $v_g'(z_B) > 0$ and the bracketed term is nonzero under range dependence, $\partial \lambda/\partial z > 0$ and $\partial \lambda/\partial y < 0$. Therefore, whenever the range effect is not exactly offset by direct prize effects, $\partial p^\star/\partial R \neq 0$.

**Step 4: Comparison with CPT and Regret Theory.**

Under CPT with fixed parameters, the evaluation of lottery $B$ is:

$$EU_B^{CPT}(p; R) = w(p)U(z_B) + (1 - w(p))U(y_B), \tag{3.24}$$

where $w(p)$ is the probability weighting function and $U$ is fixed. CPT can exhibit range effects through the curvature of $U$, but it has no term analogous to the derivative of the endogenous penalty $1/\lambda(x_r; y, z)$.

Under Regret Theory with fixed regret function $R(\cdot)$, range effects can arise through utility differences, but the regret function itself does not generate an endogenous loss-penalty attenuation channel. (Charles-Cadogan, 2016, See also) □

## 3.12 Proof of Admissible-Class Attenuation of the Rabin Calibration Implication

*Proof.* Small-stakes rejection for all $\epsilon > 0$ implies:

$$v_g(x_r + \epsilon) - \frac{1}{\lambda(x_r; x_r - \epsilon, x_r + \epsilon)} v_\ell(x_r - \epsilon) < 0 \quad \forall \epsilon > 0. \tag{3.25}$$

The small-stakes rejection premise is a local behavioral assumption. WRDU does not extrapolate it by holding a single loss-aversion parameter fixed across all stakes; the virtual loss-aversion index varies with the range of the gamble.

For the large-stakes gamble:

$$V(L_{K,G}) = \frac{1}{2} v_g(x_r + G) - \frac{1}{2\lambda(x_r; x_r - K, x_r + G)} v_\ell(x_r - K). \tag{3.26}$$

By Lemma 2, $\lambda(x_r; x_r - K, x_r + G) \to \infty$ as $G \to \infty$, so the effective loss penalty $1/\lambda$ converges to zero. Hence:

$$\lim_{G \to \infty} V(L_{K,G}) = \frac{1}{2} v_g(x_r + G) > 0. \tag{3.27}$$

Therefore sufficiently large gains are accepted. Small-stakes rejection does not imply global rejection.

In contrast, under CPT, Gul's model, Köszegi-Rabin, and Regret Theory, the relevant parameter ($\lambda$, $\beta$, or the regret function $R(\cdot)$) is fixed and does not vary with $G$. Thus

rejection of small gambles implies rejection of scaled versions in these models. (Charles-Cadogan, 2016, See also) □

### 3.13 Proof of the Admissible-Region Allais Pattern

*Proof.* The proof proceeds in four parts.

**Part 1: Existence of $D \succ C$.**

For Problem 2, the reference point is $x_r^{(2)} = 0$. All outcomes are gains:

$$V_{WRDU}(C) = 0.11 v_g(1,000,000), \tag{3.28}$$

$$V_{WRDU}(D) = 0.10 v_g(5,000,000). \tag{3.29}$$

For $D \succ C$, we need:

$$0.10 v_g(5,000,000) > 0.11 v_g(1,000,000). \tag{3.30}$$

Equivalently,

$$\frac{v_g(5,000,000)}{v_g(1,000,000)} > 1.1. \tag{3.31}$$

This inequality is a restriction on the gain utility. It is favored by lower concavity over the interval $[1M, 5M]$; it does not follow from strict concavity alone.

**Part 2: Existence of $A \succ B$.**

For Problem 1, the reference point is $x_r^{(1)} = 1,000,000 + \delta$. The WRDU values are:

$$V_{WRDU}(A) = 0, \tag{3.32}$$

$$V_{WRDU}(B) = 0.10 v_g(4,000,000 - \delta) - \frac{1}{\lambda(x_r; 1,000,000 + \delta, 4,000,000 - \delta)} v_\ell(1,000,000 + \delta). \tag{3.33}$$

For $A \succ B$, we need:

$$\lambda(x_r; 1,000,000 + \delta, 4,000,000 - \delta) < \frac{v_\ell(1,000,000 + \delta)}{0.10 v_g(4,000,000 - \delta)}. \tag{3.34}$$

Define $\Phi(\lambda) = 0.10 v_g(4,000,000 - \delta) - \frac{1}{\lambda} v_\ell(1,000,000 + \delta)$. Since $v_\ell(1,000,000 + \delta) > 0$:

$$\lim_{\lambda \to 0^+} \Phi(\lambda) = -\infty, \qquad \lim_{\lambda \to \infty} \Phi(\lambda) = 0.10 v_g(4,000,000 - \delta) > 0. \tag{3.35}$$

By the Intermediate Value Theorem, there exists $\lambda^\star > 0$ such that $\Phi(\lambda^\star) = 0$. For any $\lambda < \lambda^\star$, we have $\Phi(\lambda) < 0$, i.e., $A \succ B$.

By Lemma 2, $\lambda(x_r; y, z)$ is finite away from the boundary and range-dependent. The condition $\lambda(x_r; 1,000,000 + \delta, 4,000,000 - \delta) < \lambda^\star$ defines an open set of admissible utility functions whenever it holds strictly.

**Part 3: Joint Existence.**

The intersection of the conditions for $D \succ C$ and $A \succ B$ is non-empty. The proof is abstract and does not use a parametric utility family. Let

$$y = 1M + \delta, \qquad z = 4M - \delta.$$

For the Allais payoffs and sufficiently small $\delta \geq 0$, $y > 0.10z$.

Choose an admissible gain utility $v_g$ that is strictly increasing, strictly concave, twice continuously differentiable, and sufficiently close to the linear index on the compact interval containing $\{1M, z, 5M\}$. Such functions exist in the admissible class by taking an arbitrarily small strictly concave $C^2$ perturbation of the linear index and smoothing it near the endpoints to satisfy the WRDU boundary restrictions. By continuity,

$$\frac{v_g(5M)}{v_g(1M)} > 1.1 \qquad \text{and} \qquad \frac{0.10\, v_g(z)}{v_g'(z)} < y$$

hold for every $v_g$ in a sufficiently small $C^2$ neighborhood of that near-linear index.

Now choose any admissible loss utility $v_\ell$ that is strictly increasing, strictly concave, twice continuously differentiable, and satisfies $v_\ell(0) = 0$. Since $v_\ell'$ is strictly decreasing,

$$v_\ell(y) = \int_0^y v_\ell'(t)\, dt > y\, v_\ell'(y),$$

because $v_\ell'$ is strictly decreasing. Hence

$$\frac{v_\ell(y)}{v_\ell'(y)} > y.$$

Combining this with the previous inequality gives

$$\frac{v_\ell(y)}{v'_\ell(y)} > \frac{0.10\, v_g(z)}{v'_g(z)}.$$

Rearranging and using $\lambda(x_r; y, z) = v'_\ell(y)/v'_g(z)$ yields

$$\lambda(x_r; y, z) < \frac{v_\ell(y)}{0.10\, v_g(z)}.$$

Thus the same pair $(v_g, v_\ell)$ satisfies the $A \succ B$ condition as well as the $D \succ C$ gain-ratio condition.

Since both inequalities are strict and the maps

$$(v_g, v_\ell) \mapsto \frac{v_g(5M)}{v_g(1M)} \quad \text{and} \quad (v_g, v_\ell) \mapsto \frac{v_\ell(y)}{0.10\, v_g(z)} - \frac{v'_\ell(y)}{v'_g(z)}$$

are continuous under the $C^2$ topology wherever $v_g(1M) > 0$ and $v'_g(z) > 0$, the inequalities persist on an open $C^2$ neighborhood of the constructed pair. Therefore $\mathcal{R}_A$ is non-empty and has non-empty interior in the admissible WRDU utility class.

**Part 4: No Calibration Required.**

The conditions derived are purely abstract and require no parametric assumptions about the utility functions beyond strict monotonicity, diminishing sensitivity, and the range-dependence property of $\lambda(x_r; y, z)$. (Charles-Cadogan, 2016, See also) □

# 4 CRRA Examples: Illustrating the Admissible-Region and Admissible-Class Results

This section demonstrates how the admissible-region Allais characterization and the admissible-class Rabin attenuation result apply in practice under the commonly used constant relative risk aversion (CRRA) utility specification. All results here are illustrative only; the theoretical results in the main paper do not depend on CRRA or any other parametric form.

## 4.1 CRRA Specification

Under CRRA utility, we specify:

$$v_g(x) = \frac{x^{1-a_g}}{1-a_g}, \qquad a_g > 0, \tag{4.1}$$

$$v_\ell(x) = \frac{x^{1-a_\ell}}{1-a_\ell}, \qquad a_\ell < -1. \tag{4.2}$$

The restriction $a_\ell < -1$ ensures the desired tail properties (Charles-Cadogan, 2016, Section 4.1). We set $\alpha_g = \alpha_\ell = 1$ and $\beta_g = \beta_\ell = 0$, so that $v_g$ and $v_\ell$ are the utility functions themselves.

## 4.2 The Virtual Loss Aversion Index Under CRRA

Under the CRRA specification, the virtual loss aversion index derived from the first-order condition in the main paper is:

$$\lambda(x_r; y, z) = \frac{v_\ell'(y)}{v_g'(z)} = \frac{y^{-a_\ell}}{z^{-a_g}} = -\left(\frac{y^{-1-a_\ell}}{z^{-1-a_g}}\right)\left(\frac{a_\ell}{a_g}\right). \tag{4.3}$$

This is the closed-form expression used in the examples below.

## 4.3 Example 1: CRRA Illustration of the Admissible-Region Allais Pattern

For Problem 1, with reference point $x_r^{(1)} = 1,000,000 + \delta$:

$$V_{WRDU}(A) = 0, \tag{4.4}$$

$$V_{WRDU}(B) = 0.10 \cdot \frac{(4,000,000-\delta)^{1-a_g}}{1-a_g} - \frac{1}{\lambda(x_r; 1,000,000+\delta, 4,000,000-\delta)} \cdot \frac{(1,000,000+\delta)^{1-a_\ell}}{1-a_\ell}. \tag{4.5}$$

For $A \succ B$, we need:

$$\lambda(x_r; 1,000,000+\delta, 4,000,000-\delta) < \frac{(1,000,000+\delta)^{1-a_\ell}}{0.10(1-a_\ell)} \cdot \frac{1-a_g}{(4,000,000-\delta)^{1-a_g}}. \tag{4.6}$$

Using (4.3), this becomes:

$$\left(\frac{-a_\ell}{a_g}\right) \cdot 5,000,000^{1+a_g} < \frac{(1,000,000+\delta)^{1-a_\ell}}{0.10(1-a_\ell)} \cdot \frac{1-a_g}{(4,000,000-\delta)^{1-a_g}}. \tag{4.7}$$

For Problem 2, with reference point $x_r^{(2)} = 0$:

$$V_{WRDU}(C) = 0.11 \cdot \frac{1,000,000^{1-a_g}}{1-a_g}, \tag{4.8}$$

$$V_{WRDU}(D) = 0.10 \cdot \frac{5,000,000^{1-a_g}}{1-a_g}. \tag{4.9}$$

For $D \succ C$, we need:

$$0.10 \cdot 5,000,000^{1-a_g} > 0.11 \cdot 1,000,000^{1-a_g}, \tag{4.10}$$

which simplifies to:

$$5^{1-a_g} > 1.1, \qquad \text{or} \qquad a_g < 1 - \log_5(1.1) \approx 0.9408. \tag{4.11}$$

**Numerical Example:** Let $a_g = 0.5$, $a_\ell = -2$, $\delta = 1$. Then:

$$\lambda(x_r; 1,000,001, 3,999,999) \approx 4.472 \times 10^{10}, \tag{4.12}$$

and the threshold for $A \succ B$ is satisfied. Since $a_g = 0.5 < 0.9408$, $D \succ C$ also holds. Thus the modal Allais pattern $A \succ B$ and $D \succ C$ is obtained.

Figure 1: Admissible CRRA region for the modal Allais pattern

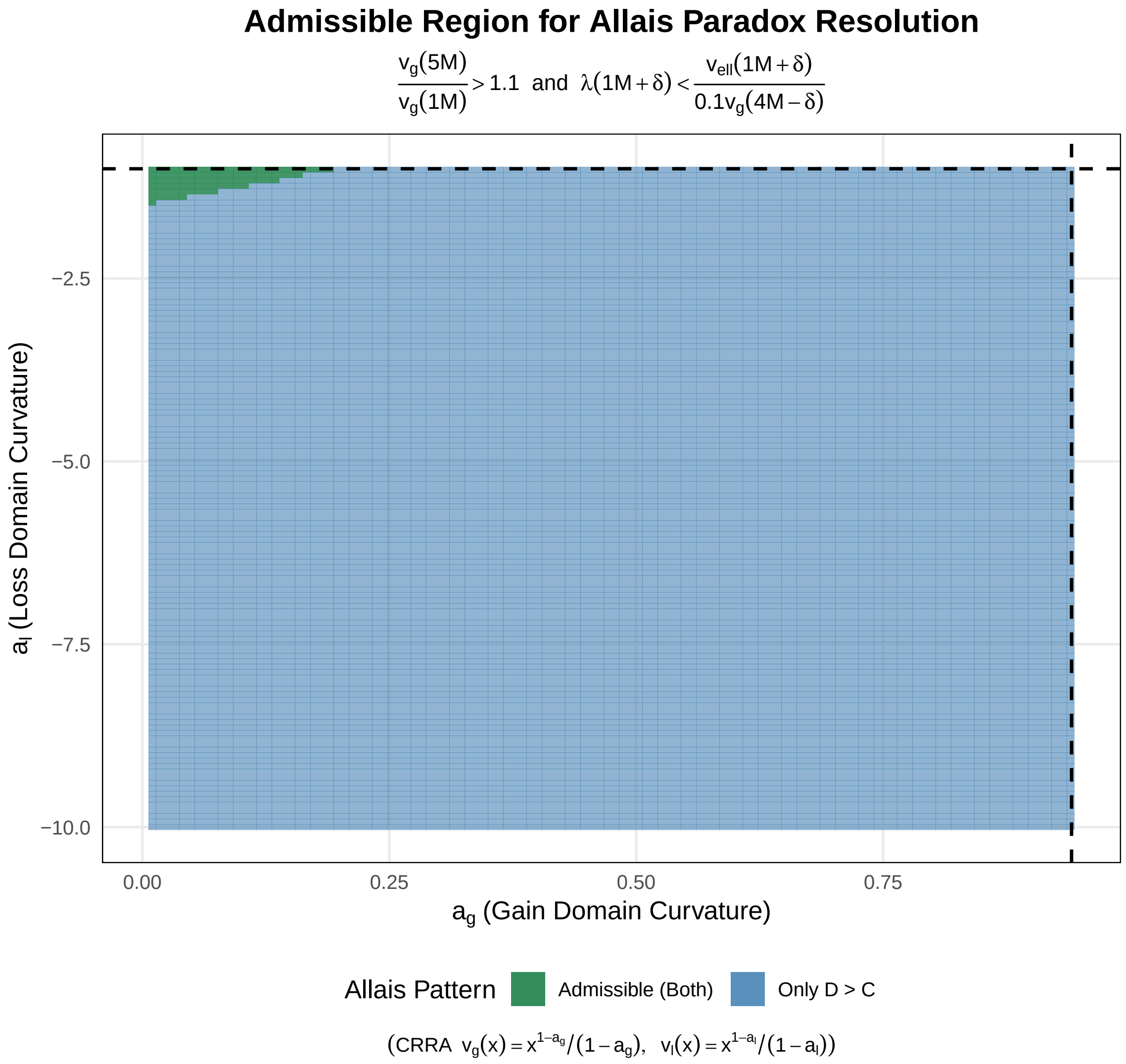


The green region satisfies both necessary inequalities for $A \succ B$ and $D \succ C$ under the CRRA illustration. The blue region satisfies only the gain-ratio condition for $D \succ C$. The dashed vertical line marks the explicit gain-ratio boundary $a_g = 1 - \log_5(1.1) \approx 0.9408$, and the dashed horizontal line marks the maintained loss-domain boundary $a_\ell = -1$.

Figure 2: Heatmap of the CRRA virtual loss-aversion index at the Allais loss-gain pair

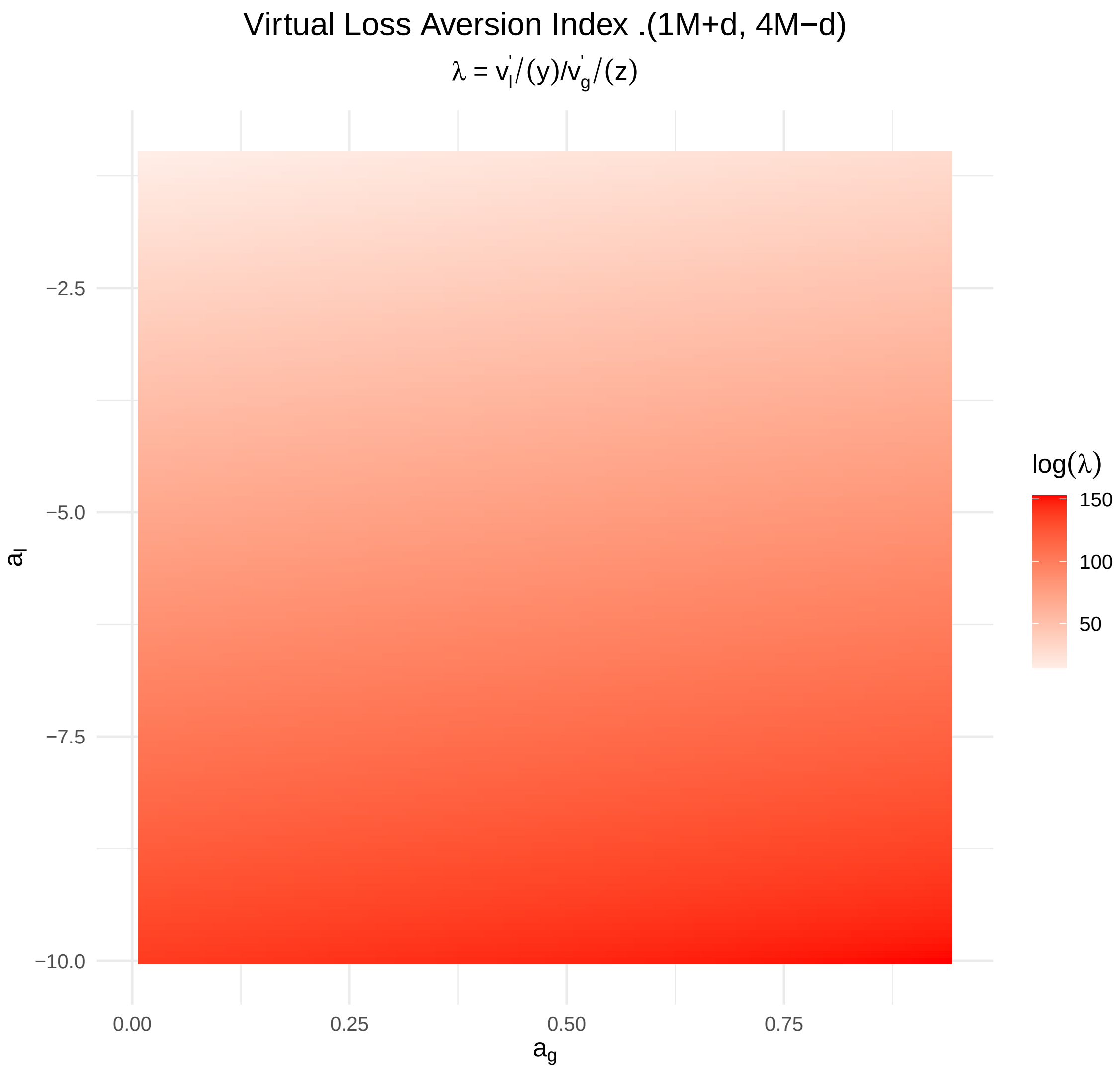


$\lambda(x_r; 1M + \delta, 4M - \delta) = v'_\ell(1M + \delta)/v'_g(4M - \delta)$. The color scale reports $\log(\lambda)$, not $\lambda$ itself. Thus a displayed value such as $\log(\lambda) = 150$ corresponds to $\lambda = \exp(150)$, an extremely large attenuation index, so the effective penalty $1/\lambda$ is correspondingly close to zero.

## 4.4 Example 2: CRRA Illustration of Rabin Attenuation

Under CRRA utility, the virtual loss aversion index for a gamble with range $(y, z)$ is:

$$\lambda(x_r; y, z) = \left(\frac{-a_\ell}{a_g}\right) y^{-1-a_\ell} z^{1+a_g}. \tag{4.13}$$

For a small symmetric gamble $L_\epsilon = (1/2, x_r + \epsilon;\ 1/2, x_r - \epsilon)$, the small-stakes evaluation is:

$$V_{WRDU}(L_\epsilon) = \frac{\epsilon}{2}\left[(x_r)^{-a_g} + \frac{(x_r)^{-a_\ell}}{\lambda(x_r; y, z)}\right] + O(\epsilon^2). \tag{4.14}$$

Rejection of $L_\epsilon$ for all small $\epsilon > 0$ requires:

$$(x_r)^{-a_g} + \frac{(x_r)^{-a_\ell}}{\lambda(x_r; y, z)} < 0. \tag{4.15}$$

For the large-stakes gamble $L_{K,G} = (1/2, x_r + G;\ 1/2, x_r - K)$:

$$V_{WRDU}(L_{K,G}) = \frac{1}{2}\frac{(x_r + G)^{1-a_g}}{1 - a_g} - \frac{1}{2\lambda(x_r; K, G)}\frac{(x_r - K)^{1-a_\ell}}{1 - a_\ell}. \tag{4.16}$$

As $G \to \infty$, $\lambda(x_r; K, G) \propto G^{1+a_g} \to \infty$, so:

$$\lim_{G\to\infty} V_{WRDU}(L_{K,G}) = \frac{1}{2}\frac{(x_r + G)^{1-a_g}}{1 - a_g} \to \infty. \tag{4.17}$$

Thus sufficiently large gains are accepted. The calibration bound vanishes:

$$\frac{G}{K} > \frac{(x_r - K)^{1-a_\ell}}{\lambda(x_r; K, G)(x_r)^{-a_g}} \to 0 \quad \text{as } G \to \infty. \tag{4.18}$$

In contrast, under CPT with fixed $\lambda$:

$$V_{CPT}(L_{K,G}) = \frac{1}{2}(x_r + G) - \frac{1}{2}\lambda(x_r - K), \tag{4.19}$$

and acceptance requires $G > \lambda K$, a fixed bound that does not vanish.

Figure 3: Rabin calibration paradox and WRDU attenuation under the CRRA illustration

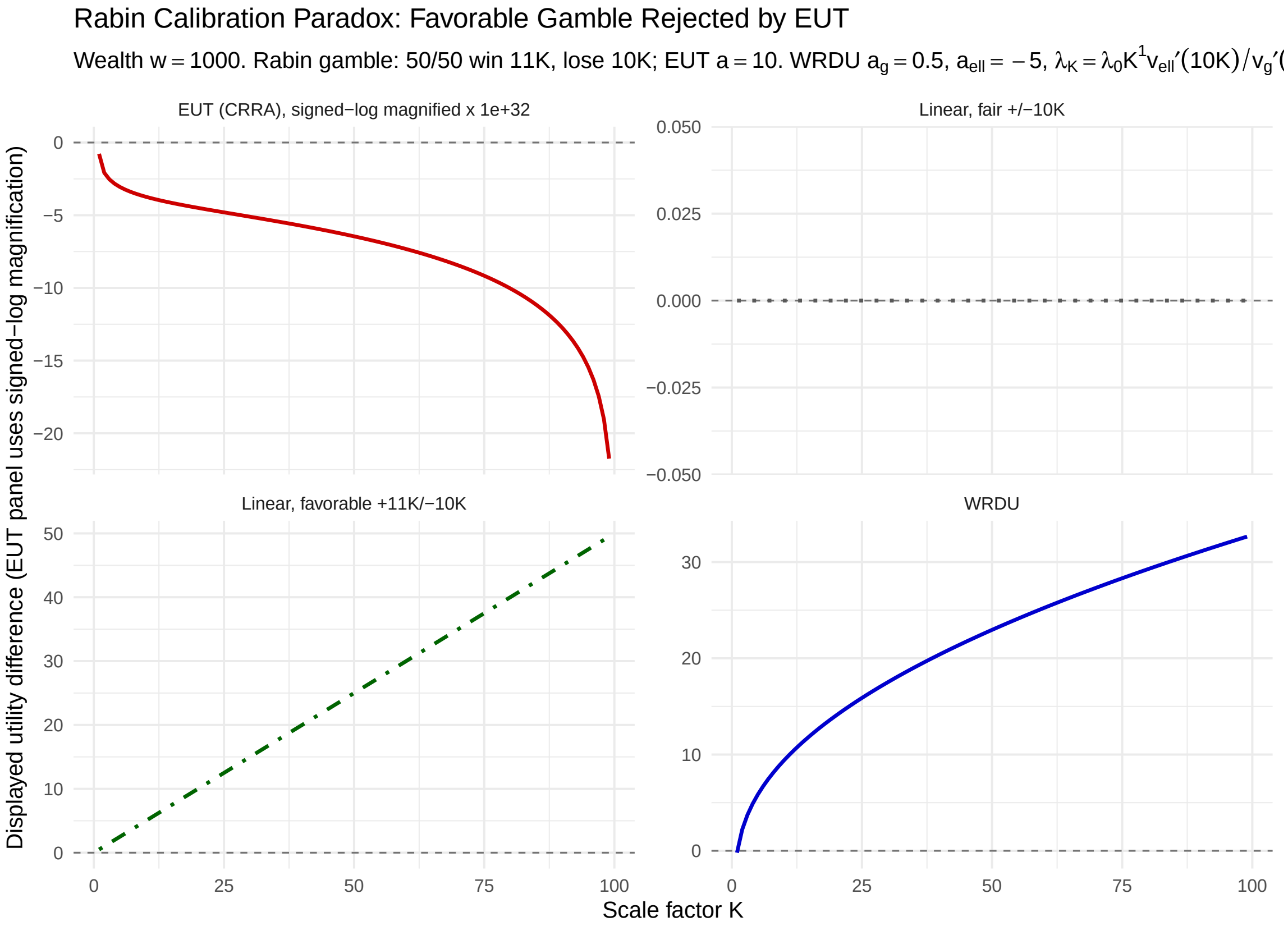


The panels report the displayed utility difference from accepting a gamble at scale $K$. Values above zero imply acceptance; values below zero imply rejection. Linear utility is exactly indifferent to the fair gamble and accepts the favorable $+11K/-10K$ gamble. The CRRA expected-utility panel is the Rabin calibration case: even though the gamble has positive expected value, sufficiently strong concavity rejects it at the small scale and continues to reject it over the plotted larger scales. Because the raw CRRA differences are extremely small in absolute value, that panel uses a sign-preserving signed-log magnification; the zero line and the sign of the evaluation are unchanged. WRDU rejects the small gamble but accepts scaled versions once the path-dependent loss-aversion index attenuates the effective loss penalty.

Figure 4: Path-dependent WRDU loss-aversion index in the Rabin calibration illustration

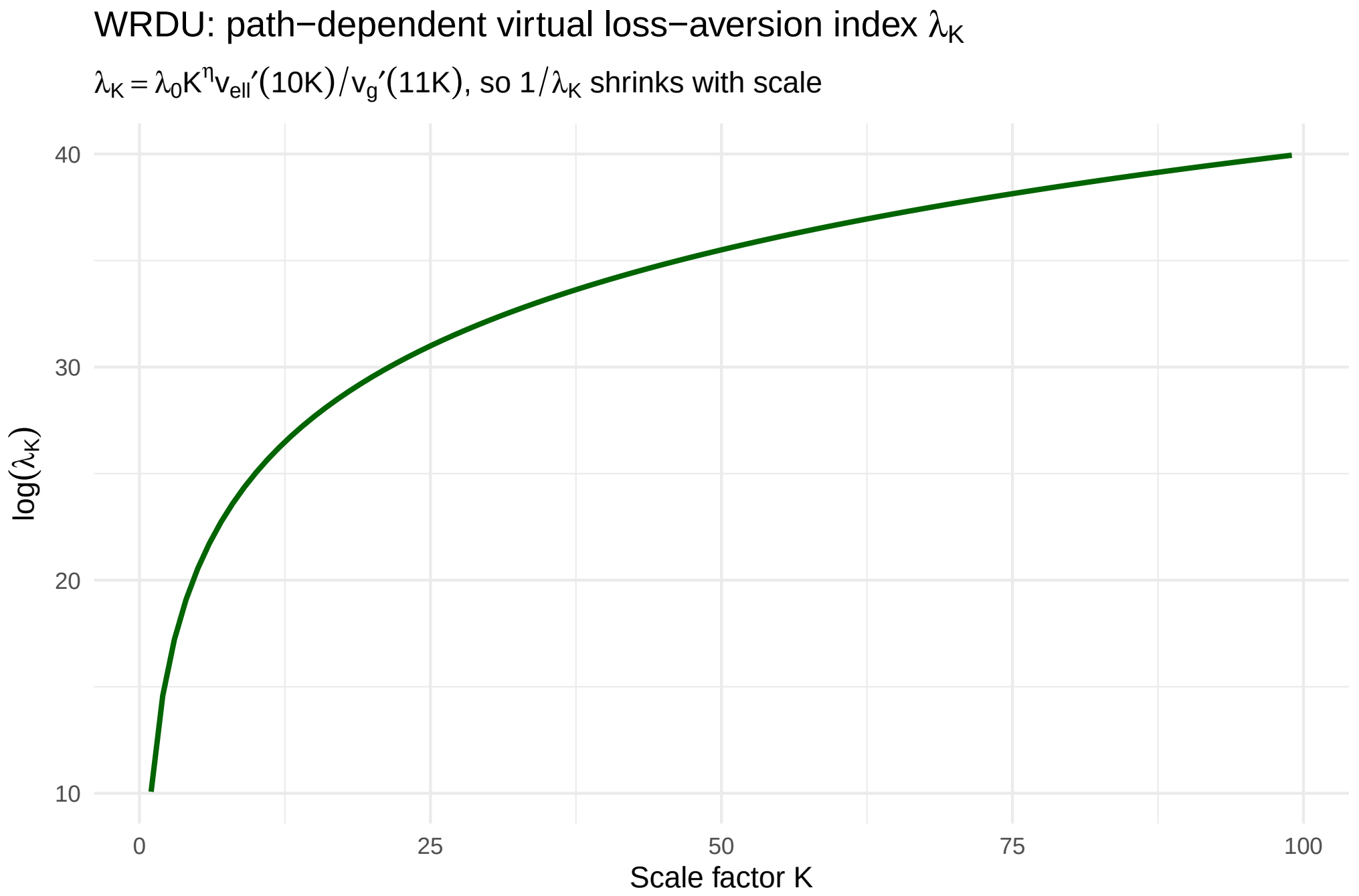


The figure plots $\log(\lambda_K)$ for the same scaled $+11K/-10K$ gamble. As $K$ rises, the path-dependent index $\lambda_K$ increases, so the Lagrangian penalty coefficient $\rho_K = 1/\lambda_K$ applied to the virtual loss term declines. This is the graphical mechanism behind the WRDU attenuation result: small-stakes rejection need not imply rejection of all larger favorable gambles.

## 4.5 Summary of CRRA Examples

The numerical examples above illustrate that:

1. The admissible-region Allais characterization (Theorem 4 in the main paper) is illustrated by CRRA utility for the parameter values $a_g = 0.5$, $a_\ell = -2$, $\delta = 1$.
2. The admissible-class Rabin attenuation result (Theorem 3 in the main paper) is illustrated by this CRRA parameterization because $\lambda(x_r; K, G) \to \infty$ as $G \to \infty$, making the calibration bound vanish in the example.
3. These examples are purely illustrative; the theoretical results in the main paper do not depend on CRRA or any other parametric assumption.

# 5 Estimation Strategy

The econometric specification follows the stochastic choice framework in Wilcox (2008, 2011), distinguishing between two standard error specifications: the **Fechner error** (utility difference noise) and the **Luce error** (probabilistic choice).

## 5.1 Fechner Error Specification

The Fechner error specification assumes that the decision maker's utility evaluation is subject to additive noise. Let the deterministic utility difference between lotteries $L_i^A$ and $L_i^B$ be:

$$\Delta V_i(\theta) = V_{WRDU}(L_i^A; \theta) - V_{WRDU}(L_i^B; \theta), \tag{5.1}$$

where $\theta = (a_g, a_\ell)$.

Under the Fechner specification, the observed utility difference is:

$$\widetilde{\Delta V}_i(\theta) = \Delta V_i(\theta) + \varepsilon_i, \tag{5.2}$$

where $\varepsilon_i \sim N(0, \sigma^2)$ is i.i.d. normal noise. The parameter $\sigma > 0$ is the **Fechner error standard deviation**, representing the magnitude of utility evaluation noise.

The likelihood function is:

$$\mathcal{L}_F(\theta, \sigma) = \prod_{i=1}^{N} \Phi\left(\frac{\Delta V_i(\theta)}{\sigma}\right)^{1\{A_i \succ B_i\}} \Phi\left(-\frac{\Delta V_i(\theta)}{\sigma}\right)^{1\{B_i \succ A_i\}}, \tag{5.3}$$

where $\Phi$ is the standard normal CDF. The ratio $\Delta V_i(\theta)/\sigma$ is the **Fechner signal-to-noise ratio**. As $\sigma \to 0$, the model converges to deterministic choice.

**Interpretation:** The Fechner error captures random variation in utility evaluation due to cognitive effort, attention, or measurement error. The parameter $\sigma$ is identified from the smoothness of the choice probability as a function of the utility difference.

## 5.2 Luce Error Specification

The Luce (1959) error specification assumes that the decision maker's choice is probabilistic according to a logit rule, with the probability of choosing $A$ over $B$ given by:

$$P(A_i \succ B_i) = \frac{\exp\big(\mu \cdot V_{WRDU}(L_i^A;\theta)\big)}{\exp\left(\mu \cdot V_{WRDU}(L_i^A;\theta)\right) + \exp\left(\mu \cdot V_{WRDU}(L_i^B;\theta)\right)}, \tag{5.4}$$

where $\mu > 0$ is the **Luce precision parameter** (or inverse temperature), representing the sensitivity of choice probabilities to utility differences. As $\mu \to \infty$, the model converges to deterministic choice; as $\mu \to 0$, choice becomes random.

Equivalently, the log-odds of choosing $A$ over $B$ is:

$$\log \frac{P(A_i \succ B_i)}{P(B_i \succ A_i)} = \mu \cdot \Delta V_i(\theta). \tag{5.5}$$

The likelihood function is:

$$\mathcal{L}_L(\theta,\mu) = \prod_{i=1}^{N} \left[\frac{\exp\big(\mu \cdot V_{WRDU}(L_i^A;\theta)\big)}{\exp\left(\mu \cdot V_{WRDU}(L_i^A;\theta)\right) + \exp\left(\mu \cdot V_{WRDU}(L_i^B;\theta)\right)}\right]^{1\{A_i \succ B_i\}} \tag{5.6}$$

$$\times \left[\frac{\exp\big(\mu \cdot V_{WRDU}(L_i^B;\theta)\big)}{\exp\left(\mu \cdot V_{WRDU}(L_i^A;\theta)\right) + \exp\left(\mu \cdot V_{WRDU}(L_i^B;\theta)\right)}\right]^{1\{B_i \succ A_i\}}. \tag{5.7}$$

**Interpretation:** The Luce error captures probabilistic choice due to random utility shocks at the decision level. The parameter $\mu$ is identified from the sensitivity of choice probabilities to utility differences. Unlike the Fechner error, the Luce error is not additive at the utility level and has a different interpretation (random utility maximization with i.i.d. extreme value errors).

## 5.3 Relationship Between Fechner and Luce Specifications

The Fechner and Luce specifications are related but not equivalent. The Luce specification corresponds to the **random utility model** with i.i.d. extreme value type I errors, while the Fechner specification corresponds to additive normal errors at the utility difference level.

A key distinction is that under the Luce specification, the utility function $V$ is only identified up to scale, while under the Fechner specification, the scale is determined by $\sigma$ (Wilcox, 2008). In practice, the two specifications often yield similar parameter estimates for the curvature parameters $(a_g, a_\ell)$ but different estimates for the precision parameter $\mu$ versus the Fechner error $\sigma$.

## 5.4 Identification

The parameters $(a_g, a_\ell)$ are identified from the curvature of the utility function in the gain and loss domains. Under the Fechner specification:

1. $a_g$ is identified from the risk aversion in gain lotteries (choices where all outcomes are gains).
2. $a_\ell$ is identified from the risk aversion in loss lotteries (choices where all outcomes are losses).
3. $\sigma$ is identified from the smoothness of the choice probability as a function of $\Delta V_i(\theta)$.

Under the Luce specification, $\mu$ and $V$ are not jointly identified without a scale normalization. We adopt the normalization $V_{WRDU}(0) = 0$, which is standard in the literature (Tversky and Kahneman, 1992).

## 5.5 Standard Errors

Standard errors for all parameters are computed via the inverse of the Fisher information matrix:

$$\mathrm{Var}(\hat{\theta}, \hat{\sigma}) = \left[ -\frac{\partial^2 \log \mathcal{L}}{\partial(\theta, \sigma) \partial(\theta, \sigma)^\top} \right]^{-1}. \tag{5.8}$$

For the implied loss aversion index $\lambda(x_r; y, z; \hat{\theta})$, standard errors are computed via the delta method:

$$\mathrm{Var}(\hat{\lambda}) = (\nabla \lambda)^\top \, \mathrm{Var}(\hat{\theta}) \, (\nabla \lambda) \,, \tag{5.9}$$

where $\nabla\lambda = \left(\frac{\partial\lambda}{\partial a_g}, \frac{\partial\lambda}{\partial a_\ell}\right)$ is evaluated at $\hat{\theta}$.

## 5.6 Model Selection

To determine which error specification is more appropriate for a given dataset, we recommend using the Akaike Information Criterion (AIC) or Bayesian Information Criterion (BIC):

$$\text{AIC} = -2\log\mathcal{L} + 2k, \tag{5.10}$$

$$\text{BIC} = -2\log\mathcal{L} + k\log N, \tag{5.11}$$

where $k$ is the number of parameters and $N$ is the number of observations.

## 5.7 Implementation Algorithm

The following algorithm summarizes the estimation procedure.

**Algorithm 1** Maximum Likelihood Estimation of WRDU Parameters

1: **Data preparation:** For each choice task $i = 1, \ldots, N$, record the two lotteries $L_i^A, L_i^B$ and the observed choice $C_i \in \{A, B\}$.

2: **Utility computation:** For given parameters $(a_g, a_\ell)$, compute $V_{WRDU}(L_i^A; \theta)$ and $V_{WRDU}(L_i^B; \theta)$ using:

$$V_{WRDU}(L) = \sum_{x \in C_g(X)} p_x \frac{x^{1-a_g}}{1-a_g} - \frac{1}{\lambda(x_r; y, z)} \sum_{x \in C_\ell(X)} p_x \frac{x^{1-a_\ell}}{1-a_\ell},$$

where $\lambda(x_r; y, z) = -\left(\frac{y^{-1-a_\ell}}{z^{-1-a_g}}\right)\left(\frac{a_\ell}{a_g}\right)$.

3: **Likelihood evaluation:**

- **Fechner:** $\mathcal{L}_F(\theta, \sigma) = \prod_{i=1}^N \Phi\left(\frac{\Delta V_i(\theta)}{\sigma}\right)^{1\{A_i \succ B_i\}} \Phi\left(-\frac{\Delta V_i(\theta)}{\sigma}\right)^{1\{B_i \succ A_i\}}$.
- **Luce:** $\mathcal{L}_L(\theta, \mu) = \prod_{i=1}^N \left[\frac{\exp(\mu V_{WRDU}(L_i^A;\theta))}{\exp(\mu V_{WRDU}(L_i^A;\theta)) + \exp(\mu V_{WRDU}(L_i^B;\theta))}\right]^{1\{A_i \succ B_i\}} \times \left[\frac{\exp(\mu V_{WRDU}(L_i^B;\theta))}{\exp(\mu V_{WRDU}(L_i^A;\theta)) + \exp(\mu V_{WRDU}(L_i^B;\theta))}\right]^{1\{B_i \succ A_i\}}$.

4: **Optimization:** Maximize $\log \mathcal{L}$ with respect to $(a_g, a_\ell, \sigma)$ (Fechner) or $(a_g, a_\ell, \mu)$ (Luce) using a constrained optimization algorithm with constraints $a_g > 0$, $a_\ell < -1$, $\sigma > 0$ (or $\mu > 0$).

5: **Standard errors:** Compute the Hessian matrix at the optimum and invert to obtain the variance-covariance matrix.

6: **Model selection:** Compare AIC/BIC for Fechner and Luce specifications and select the model with better fit.

# 6 Additional Figures

This section collects the supplementary figures used in the CRRA illustration and in the virtual-utility simulations. The first two figures report the admissible CRRA Allais region and the associated heatmap for the virtual loss-aversion index. The remaining figures present virtual utility functions for safe and risky lotteries under strong and weak rank dependence.

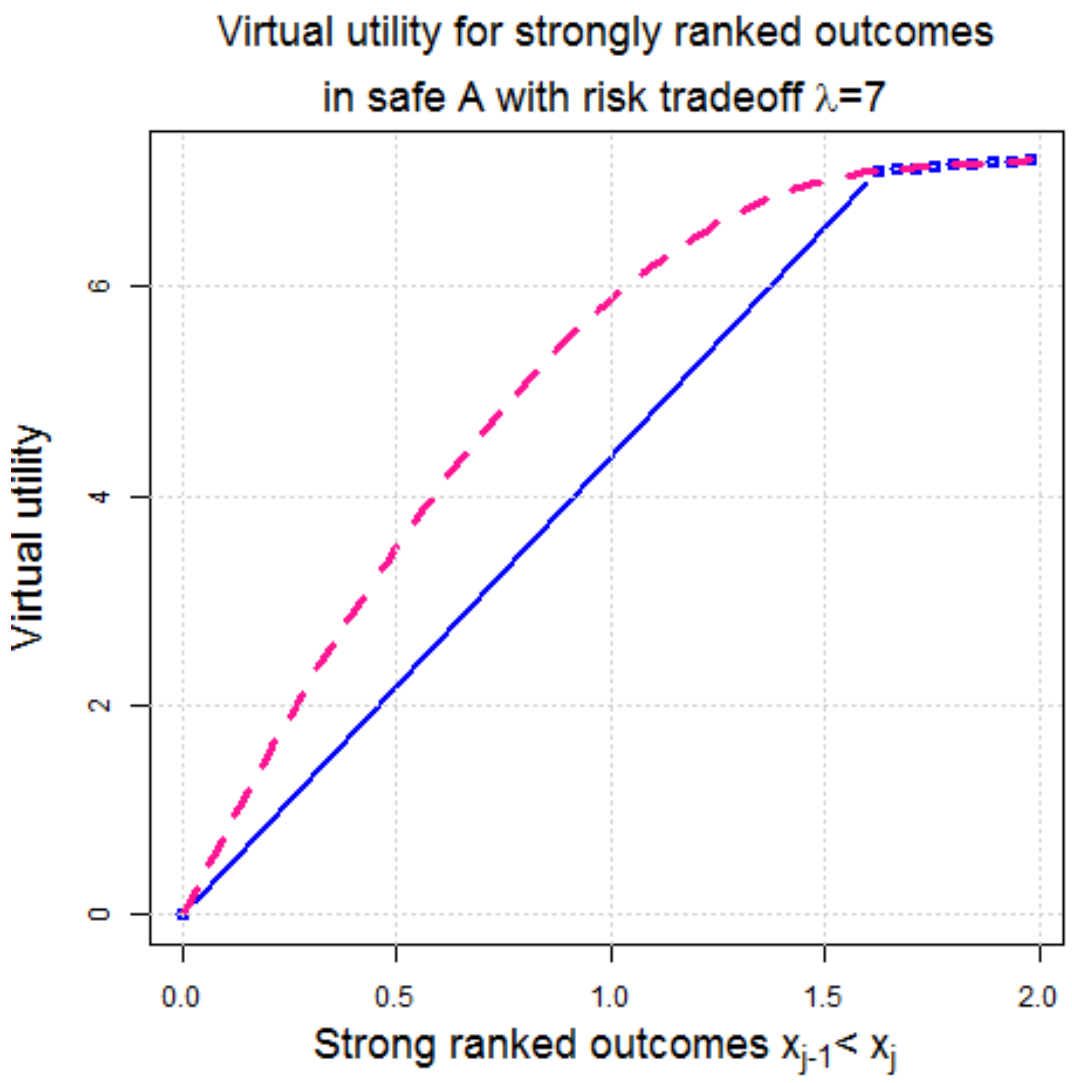


Figure 5: Virtual utility for safe A over strong dynamic reference points

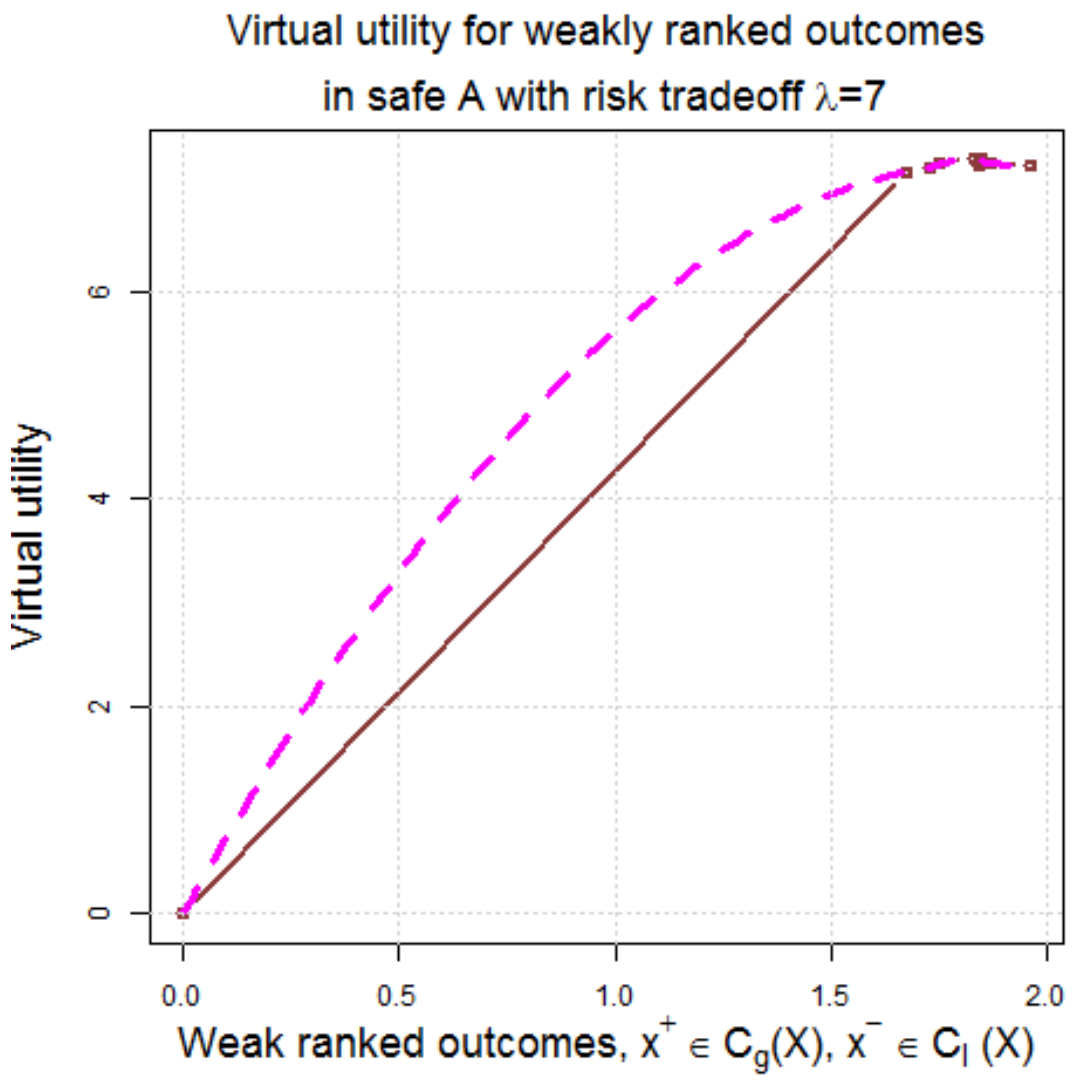


Figure 6: Virtual utility for safe A over weak dynamic reference points

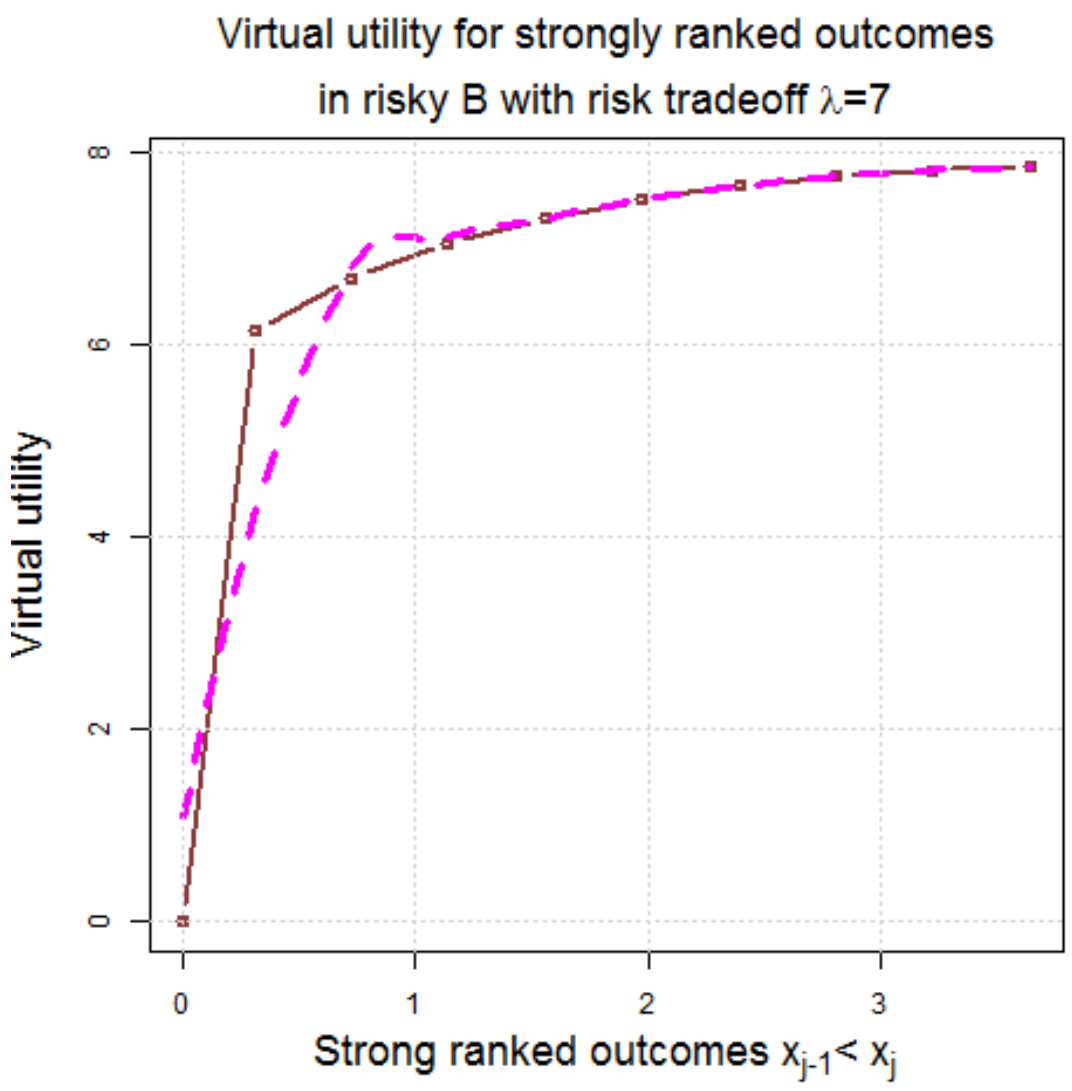


Figure 7: Virtual utility for risky B over strong dynamic reference points

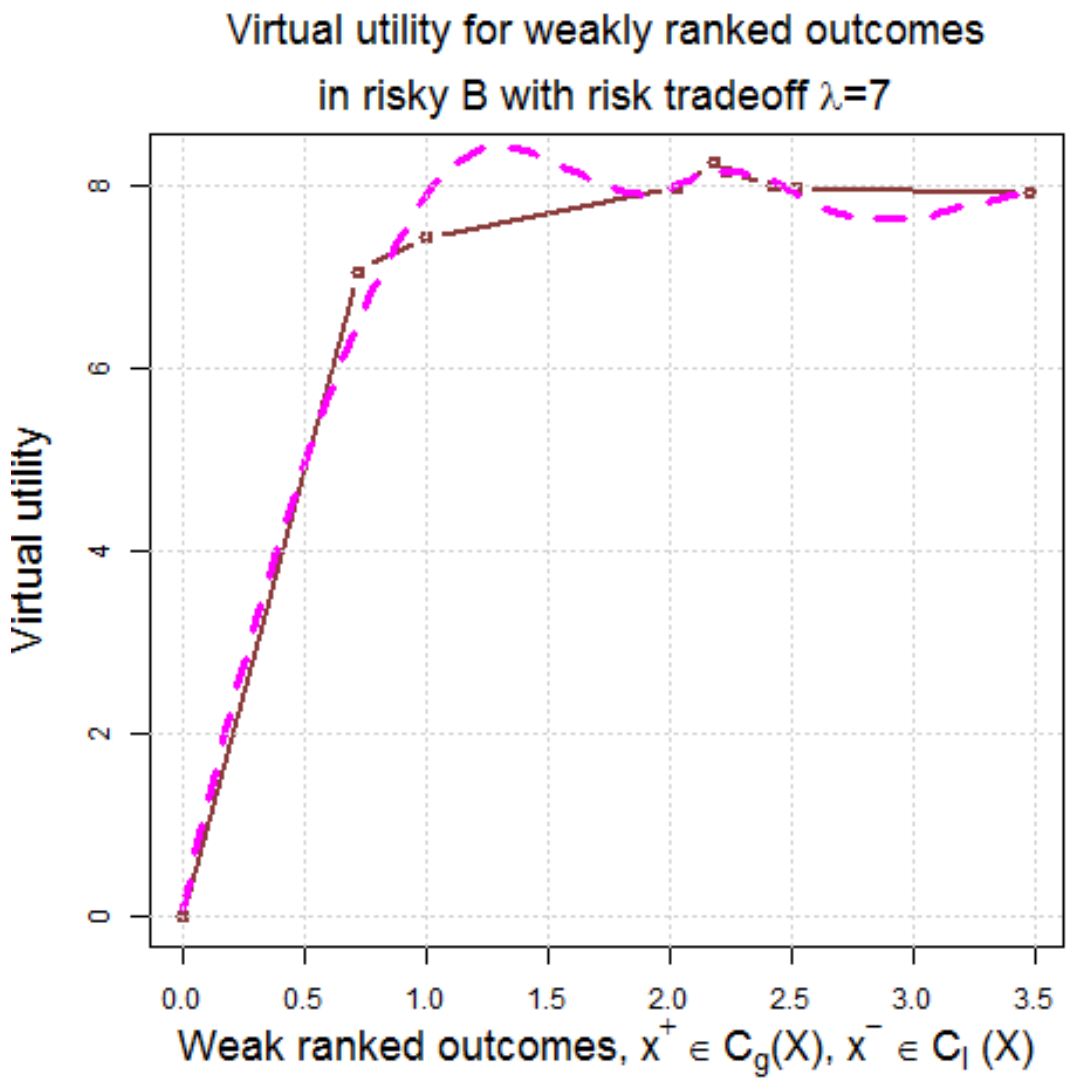


Figure 8: Virtual utility for risky B over weak dynamic reference points

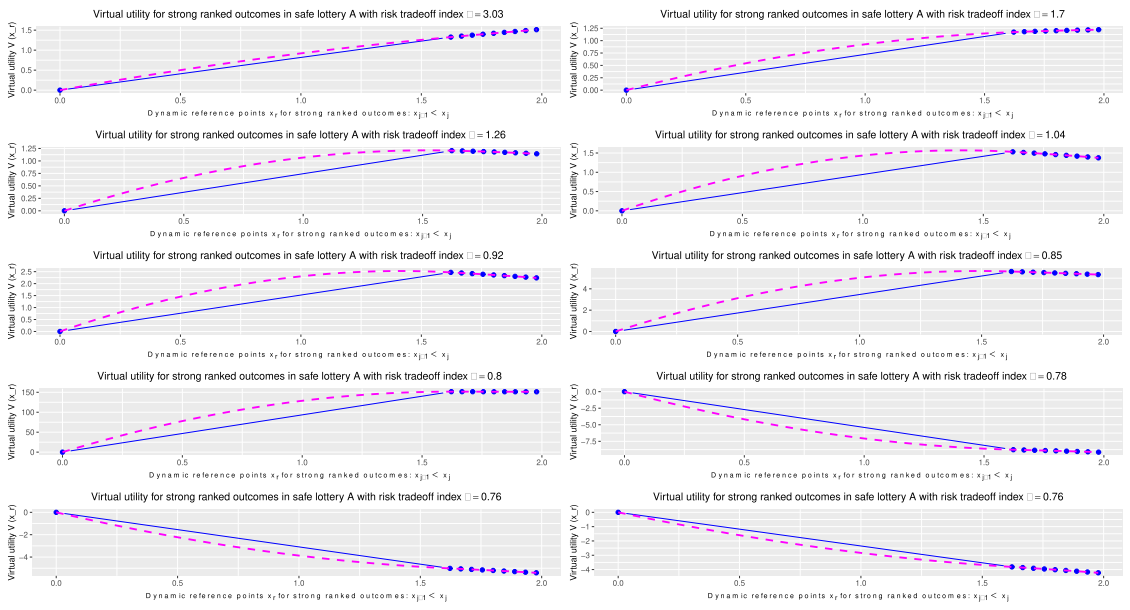

Figure 9: Virtual utility for safe A over a distribution of $\lambda$ values showing disutility region

# 7 Algorithm for Generating Virtual Utility Plots

**Algorithm 2** Generate virtual utility plots

**function** VUtil($a_\ell, a_g, x, \lambda(a_\ell, a_g)$,xdata,type)
    lenseq ← length(xdata); VUtil ← matrix(0,nrow=lenseq); $x_r$ ← VUtil
    **for** j in 1:lenseq **do**
        **if** j==1 **then**
            VUtil[1]← 0
        **else if** $j > 1$ **then**
            **if** type==0 **then**
                xdat← xdata
                $x_r[j] \leftarrow (xdat[j] + xdat[j-1])/2$
            **else if** type==1 **then**
                xdat[j-1]← $\textbf{runif}(1, xdata[1], xdata[j-1])$
                xdat[j]← $\textbf{runif}(1, xdata[j], xdata[lenseq])$
                meas = $\textbf{abs}(\textbf{jitter}((xdat[j] - xdat[j-1]), factor = 0.75))$
                $\text{lower}_{\text{bd}}[j] = \textbf{abs}(\textbf{jitter}((xdata[j-1] - xdat[j-1]), factor = 0.3))/meas$
                $\text{upper}_{\text{bd}}[j] = \textbf{abs}(\textbf{jitter}((xdata[j] - xdat[j-1]), factor = 0.3))/meas$
                $\theta = \textbf{runif}(1, \text{lower}_{\text{bd}}[j], \text{upper}_{\text{bd}}[j])$
                $x_r[j] = \textbf{jitter}(xdat[j-1] + \theta * \text{meas}, factor = .5)$
            **end if**
            VUtil[j]← $(xdat[j]^{(1-a_g)})/(1-a_g) - (1/\lambda(a_\ell, a_g)) * ((xdat[j-1]^{(1-a_\ell)})/(1-a_\ell))$
        **end if**
    **end for**
    **require**(**data**.**table**)
    VUtil← $\textbf{data.table}(\textbf{cbind}(x_r, \text{VUtil}), key = \text{“}V1\text{”})$
    **return**(VUtil)
**end function**

# 8 Empirical Evidence and Model Comparisons

This section provides the detailed empirical and model-comparison material summarized in the main paper. The purpose is to separate the core theoretical argument from supporting evidence, benchmark comparisons, and implementation details.

## 8.1 Empirical Evidence Supporting WRDU's Predictions

WRDU predicts that switching probabilities in multiple-price lists depend systematically on the range of outcomes: $\partial p^\star/\partial R \neq 0$. Csermely and Rabas (2016) find that varying the parameter that changes across rows systematically affects risk-preference distributions. Bosch-Domènech and Silvestre (2013) find that deleting items from an MPL systematically changes measured risk aversion. Habib et al. (2016) find that nonlinear transformations of lottery prices systematically change elicited risk aversion.

Peterson et al. (2021) provide broader experimental evidence for the same class of empirical regularities. Using a large risky-choice experiment and interpretable machine-learning procedures, they show that predictive performance improves when utility and probability-weighting components are allowed to depend on features of the choice environment. The most relevant connection for WRDU is not the use of machine learning itself, and not the probability-weighting component of their benchmark models. Rather, it is the finding that payoff-range information, extremal outcomes, and outcome variability help explain risky-choice behavior. WRDU takes that empirical regularity as a structural object: the gain-loss split and the Lagrangian penalty coefficient $\rho(x_r; y, z) = 1/\lambda(x_r; y, z)$ make the effective evaluation of a lottery depend on the relevant outcome geometry while leaving objective probabilities unchanged.

This comparison should be read narrowly. The evidence in Peterson et al. (2021) does not estimate the WRDU representation, does not test the derivative-ratio restriction, and does not establish the uniqueness result proved in the main paper. It does, however, weaken the empirical plausibility of globally fixed utility specifications and supports the maintained motivation that reference-sensitive and range-sensitive primitives are needed to describe risky choice. In this sense, their results provide indirect support for the empirical premise of WRDU rather than direct evidence for the full representation theorem.

WRDU also predicts that curvature asymmetry in virtual gain-loss state space drives behavior. Bocquého et al. (2023) find that the value function is significantly more linear in the loss domain than in the gain domain, while probability weighting is statistically similar across domains. This pattern is consistent with the WRDU mechanism, where the virtual loss-aversion index is generated by the ratio of marginal utilities rather than by a separate probability-weighting distortion.

Finally, WRDU predicts endogenous reference-point movement. Li and Zhong (2023) provide evidence for both endogenous and exogenous reference-point effects in intertemporal choice using multiple-price lists. In their setting, the fixed option acts as an exogenous

reference point, while past consumption acts as an endogenous reference point.

The variation in empirical loss-aversion estimates in Abdellaoui et al. (2007) and Booij et al. (2010) can be interpreted through the distinction between inner and outer measures in Charles-Cadogan (2016). In that framework,

$$\lambda = \alpha \sup_x \lambda^{\text{inner}}(x) + (1 - \alpha) \inf_x \lambda^{\text{outer}}(x), \tag{8.1}$$

where

$$\lambda^{\text{inner}}(x) = \frac{F(x)E^P[v'_\ell(x^-(x))]}{G(x)E^P[v'_g(x^+(x))]}, \qquad \lambda^{\text{outer}}(x) = \frac{F(x)v_\ell(x^-(x))}{G(x)v_g(x^+(x))}. \tag{8.2}$$

## 8.2 Exploratory Choices13k Pilot Estimation

This subsection reports a small pilot estimation using the Choices13k risky-choice data associated with Peterson et al. (2021). The exercise is included only as a diagnostic illustration of the empirical implications above. It is not a formal structural estimation exercise and should not be read as a validation of the representation theorem. The pilot uses 250 randomly sampled choice problems, treats the reported aggregate choice rate as the dependent variable, and estimates a CRRA-WRDU Luce specification with 10 multistart initial values. For each choice problem, the optimal reference point is selected by an interior grid search over the payoff support. The best pilot fit gives $a_g = 0.380$, $a_\ell = -2.449$, and $\mu = 0.091$, with a weighted root mean squared error of 0.177 for the aggregate choice rate.

The pilot is useful because it asks whether the estimated objects move in the direction predicted by WRDU. The modal reference point is close to zero, but the estimated reference point tends to decline as payoff range widens across mixed lotteries. The implied loss-aversion index $\lambda(\hat{x}_r; y, z)$ is finite for all pilot observations and increases strongly with payoff range. The Lagrangian penalty coefficient $\rho = 1/\lambda$ is largest where $\lambda < 1$ and attenuates over high-dispersion mixed lotteries. This is exactly the distinction emphasized in the main text: $\lambda$ records the virtual loss-aversion index, whereas $\rho$ is the Lagrangian penalty coefficient applied to the loss-side component.

Figure 10: Distribution of estimated optimal reference points in the Choices13k pilot

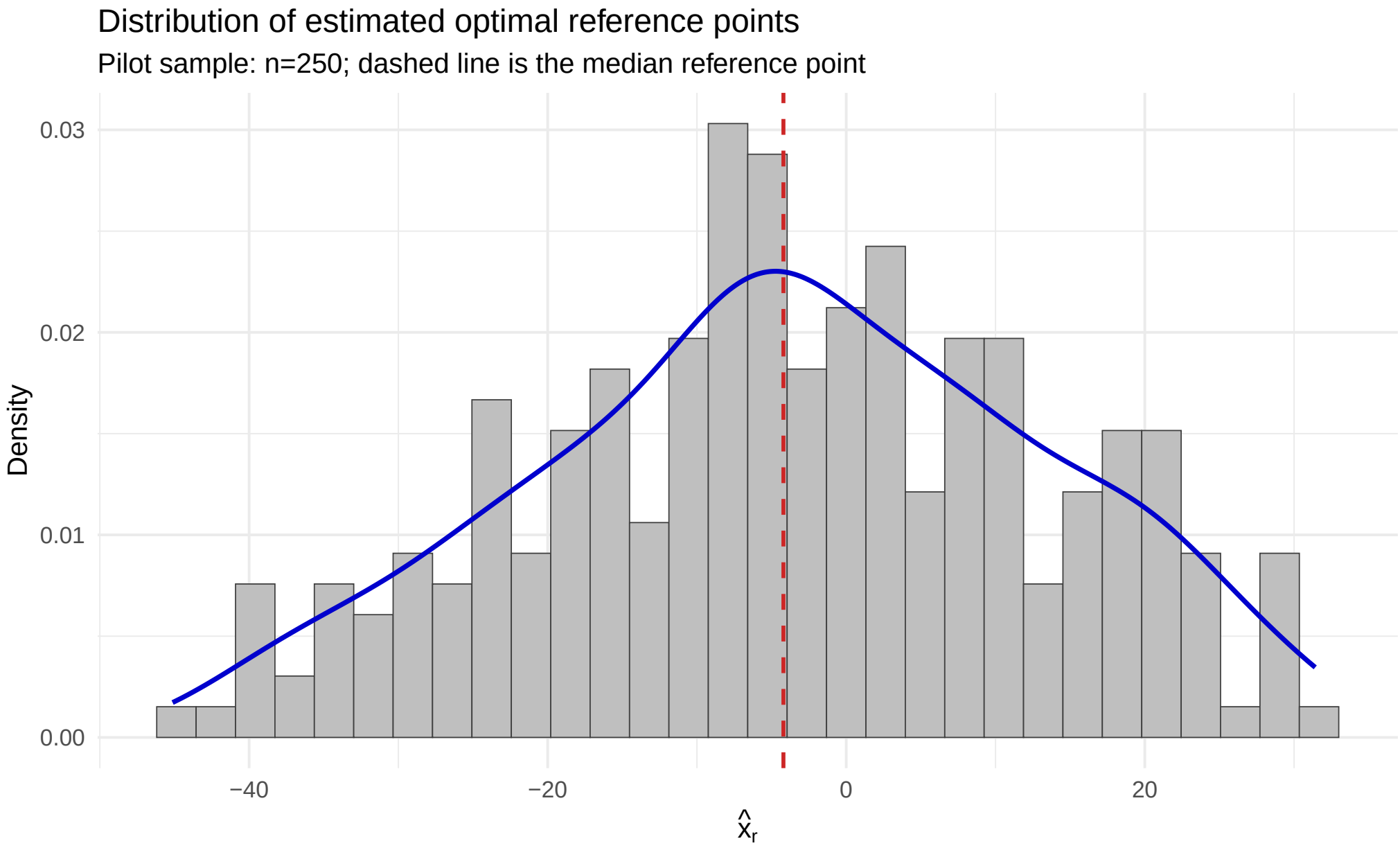


The histogram reports the distribution of $\hat{x}_r$ across the 250 pilot choice problems. The dashed vertical line marks the median reference point. The mass near zero is consistent with the interpretation of $\hat{x}_r$ as an endogenous gain-loss split rather than a fixed exogenous target.

Figure 11: Optimal reference points and payoff range in the Choices13k pilot

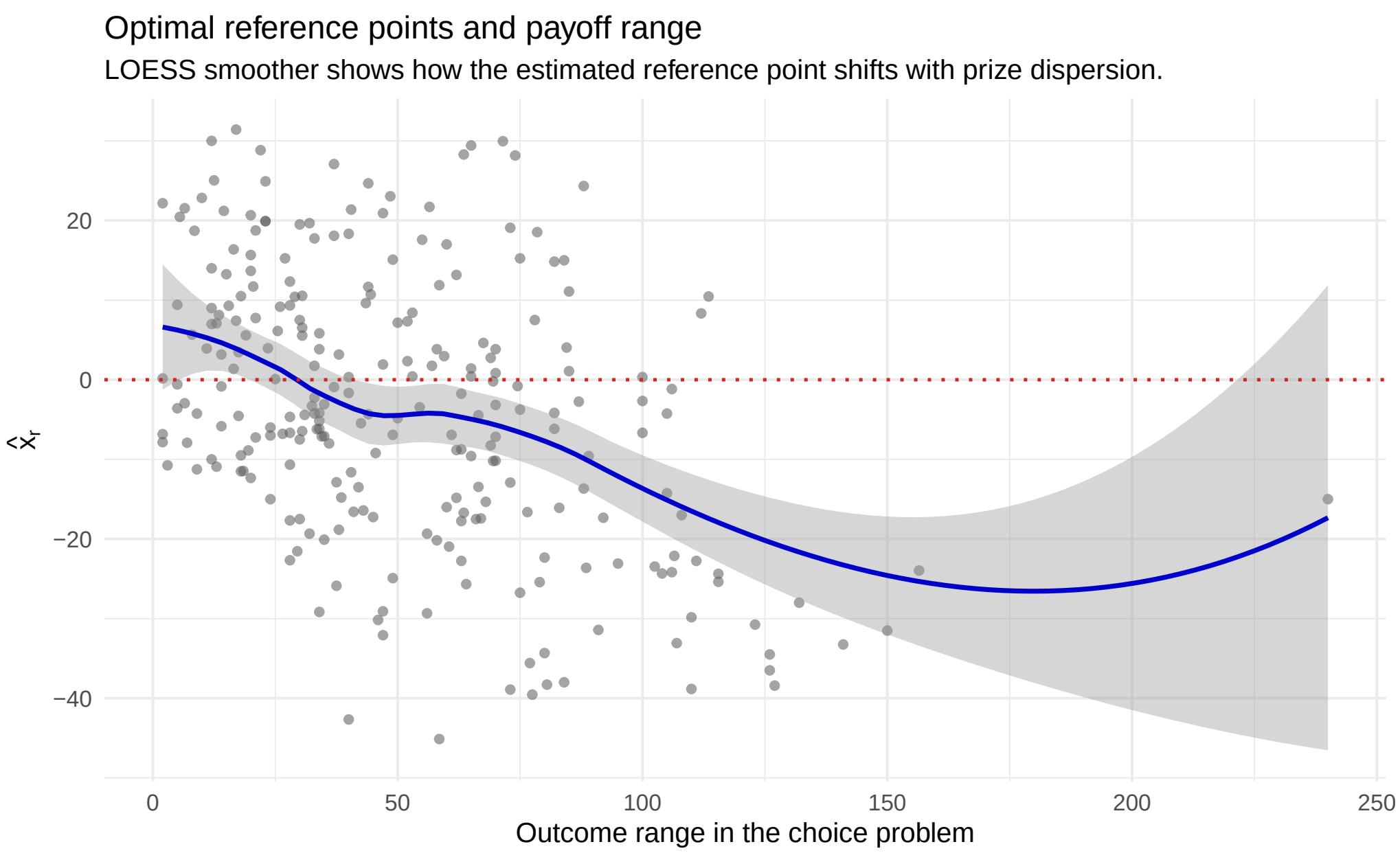


The scatter plot relates the estimated reference point $\hat{x}_r$ to the payoff range of the choice problem. The LOESS smoother declines over much of the range, indicating that the endogenous reference point shifts downward as prize dispersion increases. The wide confidence band in the far right tail reflects sparse high-range observations in the 250-row pilot.

Figure 12: Range dependence of the virtual loss-aversion index in the Choices13k pilot

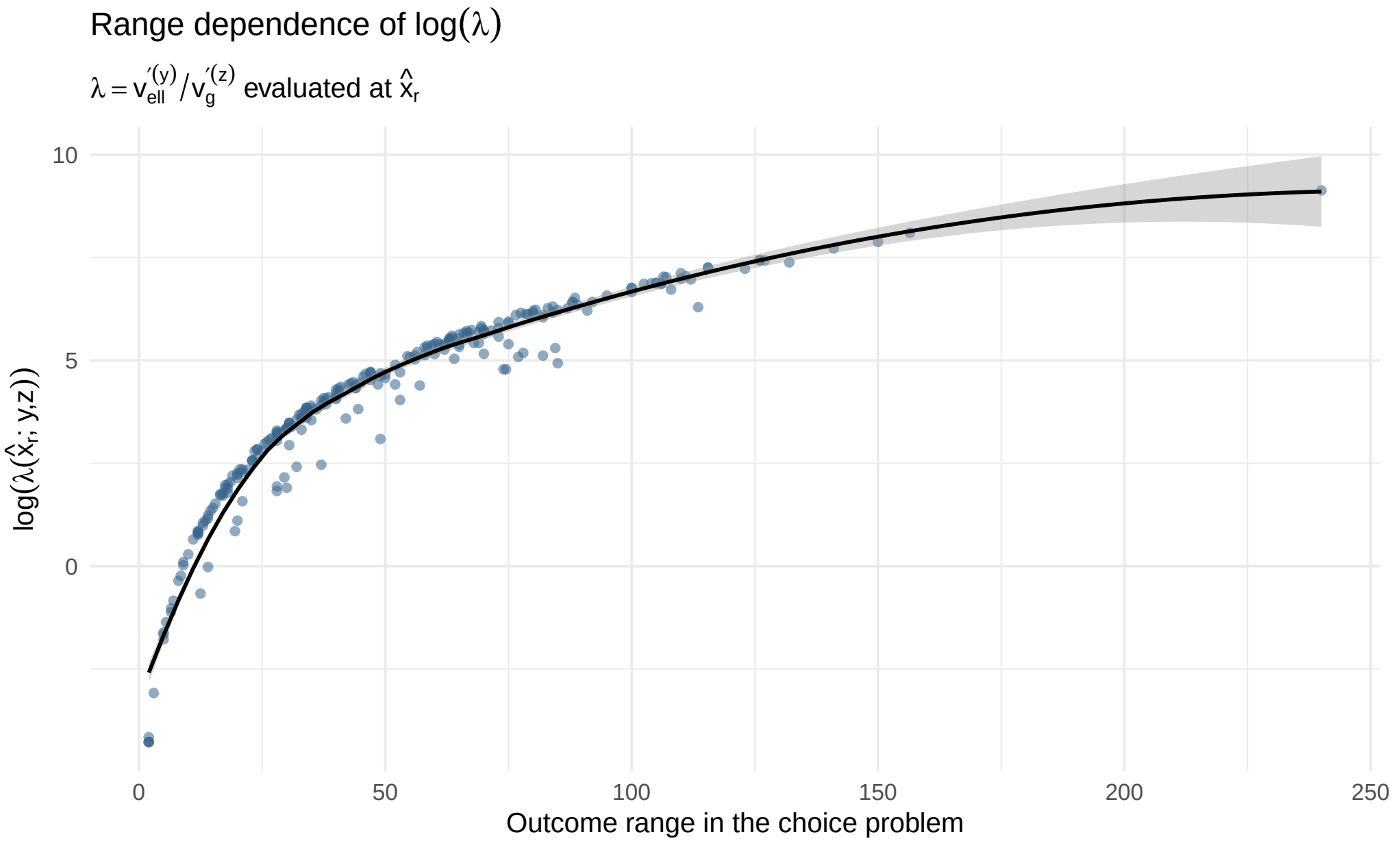


The figure plots $\log \lambda(\hat{x}_r; y, z)$ against payoff range. The positive relation is the key pilot result: once the reference point is selected endogenously, the implied virtual loss-aversion index varies systematically with outcome dispersion, as predicted by WRDU.

Figure 13: Penalty coefficient over outcome range in the Choices13k pilot

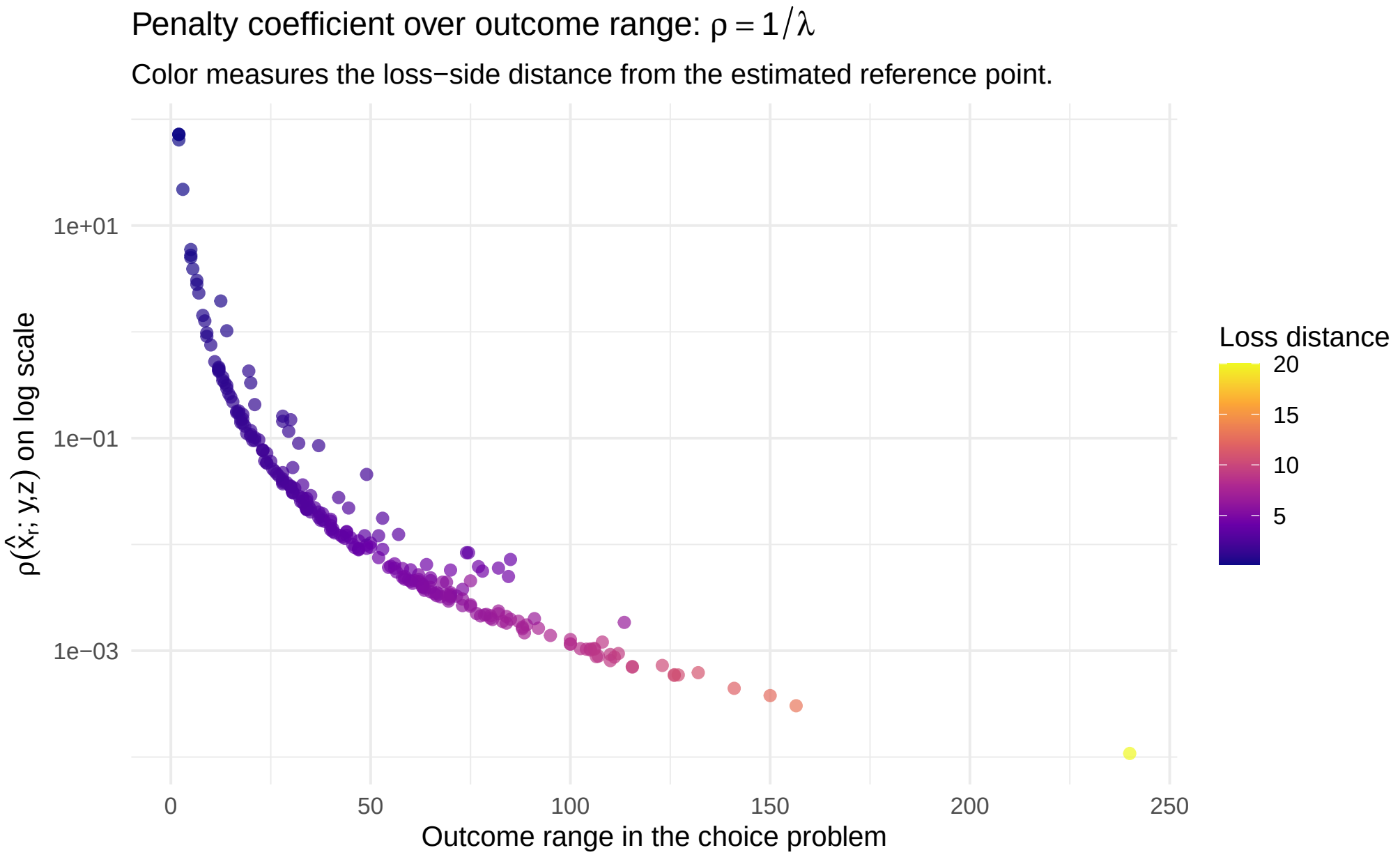


The figure plots the Lagrangian penalty coefficient $\rho = 1/\lambda$ on a log scale. Color records the loss-side distance from the estimated reference point. Since $\rho$ is the reciprocal of the loss-aversion index $\lambda$, high-dispersion mixed lotteries exhibit attenuation of the loss-side penalty coefficient even when the raw loss distance is large.

Figure 14: Aggregate choice fit in the Choices13k pilot

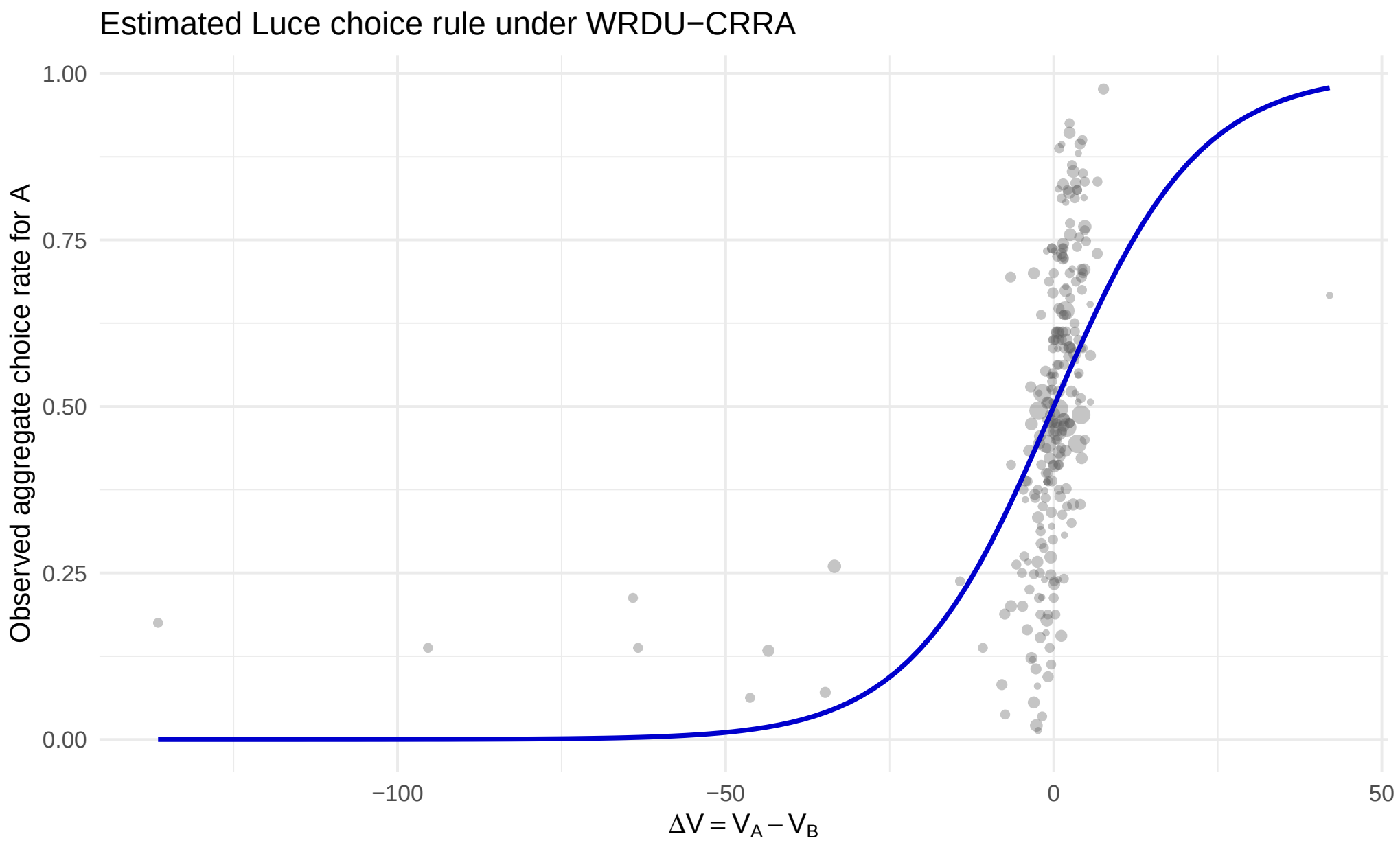


The vertical axis is the observed aggregate choice rate for Gamble A, and the horizontal axis is the estimated WRDU utility difference $V_A - V_B$. The fitted Luce curve is upward sloping, so larger WRDU utility differences correspond to higher predicted choice rates for Gamble A. The scatter should be interpreted as exploratory because the pilot uses only 250 randomly sampled problems and a coarse interior reference grid.

## 8.3 Extended Taxonomy of Related Models

The main paper gives only the comparison needed to identify the WRDU mechanism. This subsection records the fuller taxonomy. The comparison is structural rather than exhaustive: it asks whether a model changes probabilities, uses multiple utilities, introduces a fixed comparison parameter, or makes the effective gain-loss tradeoff vary with outcome dispersion.

Rank-dependent utility and dual theory transform probabilities while retaining a single utility function over outcomes (Quiggin, 1982, 1993; Yaari, 1987). Cumulative prospect theory combines probability weighting with reference dependence and a fixed loss-aversion coefficient (Tversky and Kahneman, 1992). Expectations-based reference-dependent preferences endogenize the reference point through belief consistency, but the loss-aversion parameter remains primitive (Kőszegi and Rabin, 2006; Kősegi and Rabin, 2007). Disappointment-aversion and regret models use fixed comparison terms (Gul, 1991; Bell, 1982; Loomes and Sugden, 1982). By contrast, WRDU keeps probabilities objective and derives the loss-aversion index from the first-order condition of a penalized utility problem.

A particularly close recent comparison is Cautious Utility (Cerreia-Vioglio et al., 2024). That model evaluates each lottery by the lowest monetary certainty equivalent over a set of utility functions. It therefore explains reference effects and certainty effects through uncertainty over trade-offs and cautious selection. WRDU is complementary. It does not use a set of possible utilities or a worst-certainty-equivalent rule; it uses a single normalized $C^2$ representative, a gain-loss split around an endogenous reference point, and a Lagrangian loss penalty $\rho = 1/\lambda$ that changes with the relevant loss-gain range.

Table 3: Structural Mechanisms Across Benchmark Models

| Model | Prob. distortion | Utility multiplicity | Fixed penalty | Reference point | Range-dependent penalty |
|---|---|---|---|---|---|
| EUT | No | No | No | None | No |
| RDU / Dual Theory | Yes | No | No | None | No |
| CPT | Yes | No | Yes | Exogenous | No |
| Köszegi-Rabin | No | No | Yes | Expectations equilibrium | No |
| Gul | No | No | Yes | Certainty equivalent | No |
| Regret Theory | No | No | Regret function | Pairwise comparison | No |
| Cautious Utility | No | Yes | No | Problem-dependent | Selection-based |
| **WRDU** | **No** | **No** | **No** | **Utility maximization** | **Yes** |

## 8.4 Comparison with Köszegi-Rabin Reference-Dependent Preferences

The Kőszegi and Rabin (2006); Kősegi and Rabin (2007) model is the most prominent framework that endogenizes the reference point. In that model, utility for consumption $c$ given reference point $r$ is

$$U(c|r) = m(c) + n(c|r), \tag{8.3}$$

where $m(\cdot)$ is standard consumption utility and $n(\cdot)$ is gain-loss utility:

$$n(c|r) = \begin{cases} \eta(c-r), & c \geq r, \\ \eta\lambda(c-r), & c < r, \end{cases} \tag{8.4}$$

with $\lambda > 1$ measuring loss aversion. The reference point is determined by rational expectations: in personal equilibrium, the agent's beliefs are consistent with the chosen plan.

Table 4: WRDU vs. Köszegi-Rabin

| Feature | Köszegi-Rabin | WRDU |
|---|---|---|
| Reference point formation | Expectations equilibrium | Utility maximization |
| Loss aversion parameter | Exogenous | Endogenous |
| Stake dependence | No | Yes |
| Probability transformation | No | No |
| Allais pattern | Not obtained | Obtained on $\mathcal{R}_A$ |
| Rabin implication | Not blocked | Blocked by range attenuation |
| Foundational basis | Belief consistency | Preference structure |
| Loss aversion index | Exogenous | Köbberling-Wakker index |

WRDU differs from Köszegi-Rabin in four central ways. First, WRDU derives the loss-aversion index from curvature asymmetry via the first-order condition $\lambda(x_r; y, z) = v'_\ell(y)/v'_g(z)$. Second, WRDU makes the index range-dependent, which blocks the Rabin calibration implication within the admissible WRDU class. Third, WRDU requires fewer primitive parameters because it does not impose a separate exogenous loss-aversion parameter. Fourth, the WRDU reference point emerges directly from the preference maximization problem rather than from an expectations-equilibrium fixed point.

## 8.5 Comparison with Gul's Disappointment Aversion Model and Regret Theory

WRDU shares some reference-dependent intuition with Gul (1991) and regret theory, but differs in its treatment of stake dependence. Under Gul's model, the acceptance threshold for a symmetric gamble is scale-invariant. For $L_K = (1/2, x + K\epsilon;\ 1/2, x - K\epsilon)$,

$$V_{Gul}(L_K) = \frac{1}{2}\gamma(1/2)\left[u(x + K\epsilon) + (1 + \beta)u(x - K\epsilon)\right]. \tag{8.5}$$

The disappointment parameter $\beta$ is fixed, so the model has no mechanism for the effective weight on disappointing outcomes to diminish as $K$ increases.

Under regret theory, the regret function $R(\cdot)$ is also fixed:

$$V_{Regret}(L_K) = \frac{1}{2}\left[u(x + K\epsilon) + R(u(x + K\epsilon) - u(x - K\epsilon))\right] + \frac{1}{2}\left[u(x - K\epsilon) + R(u(x - K\epsilon) - u(x + K\epsilon))\right]. \tag{8.6}$$

There is no endogenous attenuation of regret aversion as the scale of the gamble increases.

Under WRDU, by contrast,

$$V_{WRDU}(L_K) = \frac{1}{2}v_g(x + K\epsilon) - \frac{1}{2\lambda(x_r; y, x + K\epsilon)}v_\ell(x - K\epsilon). \tag{8.7}$$

As $K \to \infty$, $\lambda(x_r; y, x + K\epsilon) \to \infty$ and $1/\lambda \to 0$. The loss-side penalty attenuates with the outcome range, which is impossible in fixed-parameter disappointment-aversion and regret-theory models.

Table 5: Comprehensive Comparison: WRDU vs. Major Alternatives

| Feature | EUT | RDU | CPT | Gul | Kőszegi-Rabin | Regret | WRDU |
|---|---|---|---|---|---|---|---|
| Prob. weighting | No | Yes | Yes | Yes | No | No | No |
| Free $\lambda$ | No | No | Yes | No | Yes | No | No |
| Endogenous reference point | No | No | No | No | Yes | No | Yes |
| Loss aversion source | None | None | Exogenous | Exogenous | Exogenous | None | Endogenous |
| Range dependent | No | No | No | No | No | No | **Yes** |
| Stake dependent | No | No | No | No | No | No | **Yes** |
| Allais pattern | No | Yes | Yes | Yes | No | Partial | **Yes, on $\mathcal{R}_A$** |
| Rabin implication | No | No | No | No | No | No | **Blocked** |

# 9 References


	
\end{document}